\documentclass[english]{article}
\usepackage[T1]{fontenc}
\usepackage[latin9]{inputenc}
\usepackage{geometry}
\usepackage{color}
\usepackage{babel}
\usepackage{url}
\usepackage{amsmath}
\usepackage{amsthm}
\usepackage{amssymb}
\usepackage{graphicx}
\usepackage{esint}
\usepackage[authoryear]{natbib}
\usepackage[unicode=true,pdfusetitle,
 bookmarks=true,bookmarksnumbered=false,bookmarksopen=false,
 breaklinks=false,pdfborder={0 0 0},backref=false,colorlinks=true]
 {hyperref}
\hypersetup{
 citecolor=red,linkcolor=blue}

\makeatletter
\usepackage{tikz}
\usetikzlibrary{arrows}

\newtheorem{theorem}{Theorem}
\newtheorem{lemma}{Lemma}
\newtheorem{corollary}{Corollary}
\newtheorem{proposition}{Proposition}

\newtheorem{remark}{Remark}
\newtheorem{algo}{Algorithm}

\@ifundefined{showcaptionsetup}{}{%
 \PassOptionsToPackage{caption=false}{subfig}}
\usepackage{subfig}
\makeatother

\begin{document}

\title{Variance estimation in the particle filter}

\author{Anthony Lee and Nick Whiteley\\
University of Warwick and University of Bristol}
\maketitle
\begin{abstract}
This paper concerns numerical assessment of Monte Carlo error in particle
filters. We show that by keeping track of certain key features of
the genealogical structure arising from resampling operations, it
is possible to estimate variances of a number of standard Monte Carlo
approximations which particle filters deliver. All our estimators
can be computed from a single run of a particle filter with no further
simulation. We establish that as the number of particles grows, our
estimators are weakly consistent for asymptotic variances of the Monte
Carlo approximations and some of them are also non-asymptotically
unbiased. The asymptotic variances can be decomposed into terms corresponding
to each time step of the algorithm, and we show how to consistently
estimate each of these terms. When the number of particles may vary
over time, this allows approximation of the asymptotically optimal
allocation of particle numbers. 
\end{abstract}

\section{Introduction}

Particle filters, or sequential Monte Carlo methods, provide Monte
Carlo approximations of integrals with respect to sequences of measures.
In popular statistical inference applications, these measures arise
naturally from conditional distributions in hidden Markov models,
or are constructed artificially to bridge between target distributions
in Bayesian analysis. The numbers of particles used to perform the
approximation controls the tradeoff between computational complexity
and accuracy. Theoretical properties of this relationship have been
the subject of intensive research; the literature includes a number
of central limit theorems \citep{del1999central,Chopin2004,kunsch2005recursive,Douc08weighted}
and a variety of refined asymptotic \citep{douc2005moderate,del2007sharp}
and non-asymptotic \citep{del2001genealogies,cerou2011nonasymptotic}
results. These studies provide a wealth of insight into the mathematical
behaviour of particle filter approximations and validate them theoretically,
but considerably less is known about how, in practice, to extract
information from a realization of a single particle filter in order
to report numerical measures of Monte Carlo error. This is in notable
contrast to other families of Monte Carlo techniques, especially Markov
chain Monte Carlo, for which an extensive literature on variance estimation
exists. Our main aim is to address this gap.

We introduce particle filters via a framework of Feynman--Kac models
\citep{DelMoral2004}. This approach allows us to identify the key
generic ingredients defining particle filters and the measures they
approximate, and emphasizes that our variance estimators can be used
across application areas. Based on a single realization of a particle
filter, we provide unbiased estimators of the variance and individual
asymptotic variance terms for a class of unnormalized particle approximations.
No estimators of these quantities based on a single run of a particle
filter have previously appeared in the literature. Upon suitable rescaling,
we establish that our estimators are weakly consistent for asymptotic
variances associated with a larger class of particle approximations.
One of these re-scaled estimators is closely related to that proposed
by \citet{chan2013general}, which is the only other consistent asymptotic
variance estimator based on a single realization of a particle filter
in the literature. We also demonstrate how one can use the estimators
to inform the choice of algorithm parameters in an attempt to improve
performance.

\section{Particle filters}

\subsection{Notation and conventions\label{sub:Notation-and-conventions}}

For a generic measurable space $(\mathsf{E},\mathcal{E})$, we denote
by $\mathcal{L}(\mathcal{E})$ the set of $\mathbb{R}$-valued, $\mathcal{E}$-measurable
and bounded functions on $\mathsf{E}$. For $\varphi\in\mathcal{L}(\mathcal{E})$,
$\mu$ a measure and $K$ an integral kernel on $(\mathsf{E},\mathcal{E})$,
we write $\mu(\varphi)=\int_{\mathsf{E}}\varphi(x)\mu({\rm d}x)$,
$K(\varphi)(x)=\int_{\mathsf{E}}K(x,{\rm d}x^{\prime})\varphi(x^{\prime})$
and $\mu K(A)=\int_{\mathsf{E}}\mu({\rm d}x)K(x,A)$. Constant functions
$x\in E\mapsto c\in\mathbb{R}$ are denoted simply by $c$. For $\varphi\in\mathcal{L}(\mathcal{E})$,
$\varphi^{\otimes2}(x,x^{\prime})=\varphi(x)\varphi(x^{\prime})$.
The Dirac measure located at $x$ is denoted $\delta_{x}$. For any
sequence $(a_{n})_{n\in\mathbb{Z}}$ and $p\leq q$, $a_{p:q}=(a_{p},\ldots,a_{q})$.
For any $m\in\mathbb{N}$, $[m]=\{1,\ldots,m\}$. For a vector of
positive values $(a_{1},\ldots,a_{m})$, we denote by $\mathcal{C}(a_{1},\ldots,a_{m})$
the Categorical distribution over $\{1,\ldots,m\}$ with probabilities
$(a_{1}/\sum_{i=1}^{m}a_{i},\ldots,a_{m}/\sum_{i=1}^{m}a_{i})$. When
a random variable is indexed by a superscript $N$, a sequence of
such random variables is implicitly defined by considering each value
$N\in\mathbb{N}$, and limits will always be taken along this sequence.

\subsection{Discrete time Feynman--Kac models\label{sub:Discrete-time-Feynman--Kac}}

On a measurable space $\left(\mathsf{X},\mathcal{X}\right)$ with
$n$ a non-negative integer, let $M_{0}$ be a probability measure,
$M_{1},\ldots,M_{n}$ a sequence of Markov kernels and $G_{0},\ldots,G_{n}$
a sequence of $\mathbb{R}$-valued, strictly positive, upper-bounded
functions. We assume throughout that $\mathsf{X}$ does not consist
of a single point. We define a sequence of measures by $\gamma_{0}=M_{0}$
and, recursively,
\begin{equation}
\gamma_{p}(S)=\int_{\mathsf{X}}\gamma_{p-1}({\rm d}x)G_{p-1}(x)M_{p}(x,S),\qquad p\in[n],\quad S\in\mathcal{X}.\label{eq:gamman}
\end{equation}

\noindent Since $\gamma_{p}(\mathsf{X})\in(0,\infty)$ for each $p$,
the following probability measures are well-defined: 
\begin{equation}
\eta_{p}(S)=\frac{\gamma_{p}(S)}{\gamma_{p}(\mathsf{X})},\qquad p\in\{0,\ldots,n\},\quad S\in\mathcal{X}.\label{eq:etan}
\end{equation}
The representation 
\begin{equation}
\gamma_{n}(\varphi)=E\left\{ \varphi(X_{n})\prod_{p=0}^{n-1}G_{p}(X_{p})\right\} ,\label{eq:gammanphi_expectation}
\end{equation}
where the expectation is taken with respect to the Markov chain with
initial distribution $X_{0}\sim M_{0}$ and transitions $X_{p}\sim M_{p}(X_{p-1},\cdot)$,
establishes the connection to Feynman--Kac formulae. Measures with
the structure in (\ref{eq:gamman})--(\ref{eq:etan}) arise in a variety
of statistical contexts.

\subsection{Motivating examples of Feynman--Kac models\label{sub:Statistical-examples-of}}

As a first example, consider a hidden Markov model: a bivariate Markov
chain $(X_{p},Y_{p})_{p=0,\ldots,n}$ where $(X_{p})_{p=0,\ldots,n}$
is itself Markov with initial distribution $M_{0}$ and transitions
$X_{p}\sim M_{p}(X_{p-1},\cdot)$, and such that each $Y_{p}$ is
conditionally independent of $(X_{q},Y_{q};q\neq p)$ given $X_{p}$.
If the conditional distribution of $Y_{p}$ given $X_{p}$ admits
a density $g_{p}(X_{p},\cdot)$ and one fixes a sequence of observed
values $y_{0},\ldots,y_{n-1}$, then with $G_{p}(x_{p})=g_{p}(x_{p},y_{p})$,
$\eta_{n}$ is the conditional distribution of $X_{n}$ given $y_{0},\ldots,y_{n-1}$.
Hence, $\eta_{n}(\varphi)$ is a conditional expectation and $\gamma_{n}(\mathsf{X})=\gamma_{n}(1)$
is the marginal likelihood of $y_{0},\ldots,y_{n-1}$

As a second example, consider the following sequential simulation
setup. Let $\pi_{0}$ and $\pi_{1}$ be two probability measures on
$\left(\mathsf{X},\mathcal{X}\right)$ such that $\pi_{0}(dx)=\bar{\pi}_{0}(x)dx/Z_{0}$
and $\pi_{1}(dx)=\bar{\pi}_{1}(x)dx/Z_{1}$, where $\bar{\pi}_{0}$
and $\bar{\pi}_{1}$ are unnormalized probability densities with respect
to a common dominating measure $dx$ and $Z_{i}=\int_{\mathsf{X}}\bar{\pi}_{i}(x)dx$,
$i\in\{0,1\}$ are integrals unavailable in closed form. In Bayesian
statistics $\pi_{1}$ may arise as a posterior distribution from which
one wishes to sample, e.g. having multiple modes and complicated local
covariance structures, $\pi_{0}$ is a more benign distribution from
which sampling is feasible, and calculating $Z_{1}/Z_{0}$ allows
assessment of model fit. Introducing a sequence $0=\beta_{0}<\cdots<\beta_{n}=1$
and taking $G_{p}(x)=\left\{ \bar{\pi}_{1}(x)/\bar{\pi}_{0}(x)\right\} ^{\beta_{p+1}-\beta_{p}}$,
$M_{0}=\pi_{0}$ and for each $p=1,\ldots,n$, $M_{p}$ as a Markov
kernel invariant with respect to the distribution with density proportional
to $\bar{\pi}_{0}(x)^{1-\beta_{p}}\bar{\pi}_{1}(x)^{\beta_{p}}$,
one obtains by elementary manipulations
\[
\gamma_{p}(S)=\frac{1}{Z_{0}}\int_{S}\bar{\pi}_{0}(x)^{1-\beta_{p}}\bar{\pi}_{1}(x)^{\beta_{p}}dx,\quad\eta_{n}=\pi_{1},\quad\gamma_{n}(\mathsf{X})=\frac{Z_{1}}{Z_{0}},
\]
so that $\eta_{1},\ldots,\eta_{n-1}$ forms a sequence of intermediate
distributions between $\pi_{0}$ and $\pi_{1}$. This type of construction
appears in \citep{DelMoral2006} and references therein.

\subsection{Particle approximations\label{sub:Particle-filters-and-approx}}

We now introduce particle approximations of the measures in (\ref{eq:gamman})--(\ref{eq:etan}).
Let $c_{0:n}$ be a sequence of positive real numbers and $N\in\mathbb{N}$.
We define a sequence of particle numbers $N_{0:n}$ by $N_{p}=\lceil c_{p}N\rceil$
for $p\in\{0,\ldots,n\}$. To avoid notational complications, we shall
assume throughout that $c_{0:n}$ and $N$ are such that $\text{min}_{p}N_{p}\geq2$.
The particle system consists of a sequence $\zeta=\zeta_{0:n}$, where
for each $p$, $\zeta_{p}=(\zeta_{p}^{1},\ldots,\zeta_{p}^{N_{p}})$
and each $\zeta_{p}^{i}$ is valued in $\mathsf{X}$. To describe
the resampling operation we also introduce random variables denoting
the indices of the ancestors of each random variable $\zeta_{p}^{i}$.
That is, for each $i\in[N_{p}]$, $A_{p-1}^{i}$ is a $[N_{p-1}]$-valued
random variable and we write $A_{p-1}=(A_{p-1}^{1},\ldots,A_{p-1}^{N_{p}})$
for $p\in[n]$ and $A=A_{0:n-1}$. 

A simple algorithmic description of the particle system is given in
Algorithm~\ref{alg:bpf}. An important and non-standard feature here
is that we keep track of a collection of Eve indices $E_{0:n}$ with
$E_{p}=(E_{p}^{1},\ldots,E_{p}^{N_{p}})$ for each $p$, which will
be put to use in our variance estimators. We adopt the Eve terminology
because $E_{p}^{i}$ represents the index of the time $0$ ancestor
of $\zeta_{p}^{i}$ . The fact that $N_{p}$ may vary with $p$ is
also atypical, and allows us to address asymptotically optimal particle
allocation in Section~\ref{sub:approx_allocation}. On a first reading,
one may wish to assume that $N_{0:n}$ is not time-varying, i.e. $c_{p}=1$
so $N_{p}=N$ for all $p\in\{0,\ldots,n\}$. Figure~\ref{fig:A-particle-system}
is a graphical representation of a realization of a small particle
system.

\begin{algo}The particle filter.\label{alg:bpf}
\begin{enumerate}
\item At time $0$: for each $i\in[N_{0}]$, sample $\zeta_{0}^{i}\sim M_{0}(\cdot)$
and set $E_{0}^{i}\leftarrow i$.
\item \begin{flushleft}
At each time $p=1,\ldots,n$: for each $i\in[N_{p}]$,
\par\end{flushleft}

\begin{enumerate}
\item Sample $A_{p-1}^{i}\sim\mathcal{C}\left\{ G_{p-1}(\zeta_{p-1}^{1}),\ldots,G_{p-1}(\zeta_{p-1}^{N_{p-1}})\right\} $.
\item Sample $\zeta_{p}^{i}\sim M_{p}(\zeta_{p-1}^{A_{p-1}^{i}},\cdot)$
and set $E_{p}^{i}\leftarrow E_{p-1}^{A_{p-1}^{i}}$.
\end{enumerate}
\end{enumerate}
\end{algo}

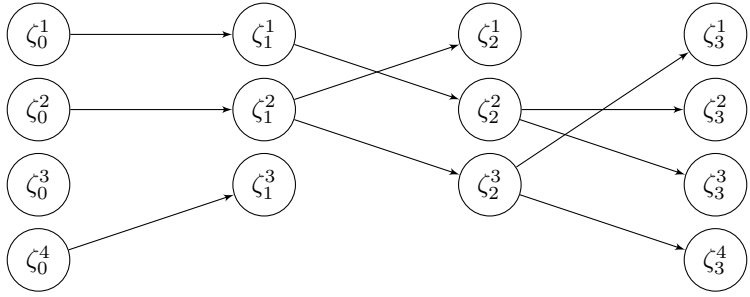
\begin{figure}
\centering{}\begin{tikzpicture}
\tikzset{vertex/.style = {shape=circle,draw,minimum size=2em}}
\tikzset{edge/.style = {->,> = latex'}}
\node[vertex] (z01) at  (0,0) {$\zeta_0^1$};
\node[vertex] (z02) at  (0,-1) {$\zeta_0^2$};
\node[vertex] (z03) at  (0,-2) {$\zeta_0^3$};
\node[vertex] (z04) at  (0,-3) {$\zeta_0^4$};
\node[vertex] (z11) at  (3,0) {$\zeta_1^1$};
\node[vertex] (z12) at  (3,-1) {$\zeta_1^2$};
\node[vertex] (z13) at  (3,-2) {$\zeta_1^3$};
\node[vertex] (z21) at  (6,0) {$\zeta_2^1$};
\node[vertex] (z22) at  (6,-1) {$\zeta_2^2$};
\node[vertex] (z23) at  (6,-2) {$\zeta_2^3$};
\node[vertex] (z31) at  (9,0) {$\zeta_3^1$};
\node[vertex] (z32) at  (9,-1) {$\zeta_3^2$};
\node[vertex] (z33) at  (9,-2) {$\zeta_3^3$};
\node[vertex] (z34) at  (9,-3) {$\zeta_3^4$};
\draw[edge] (z01) to (z11);
\draw[edge] (z02) to (z12); 
\draw[edge] (z04) to (z13);
\draw[edge] (z12) to (z21); 
\draw[edge] (z11) to (z22);
\draw[edge] (z12) to (z23);
\draw[edge] (z23) to (z31); 
\draw[edge] (z22) to (z32); 
\draw[edge] (z22) to (z33);
\draw[edge] (z23) to (z34);
\end{tikzpicture}\caption{\label{fig:A-particle-system}A particle system with $n=3$ and $N_{0:3}=(4,3,3,4)$.
An arrow from $\zeta_{p-1}^{i}$ to $\zeta_{p}^{j}$ indicates that
the ancestor of $\zeta_{p}^{j}$ is $\zeta_{p-1}^{i}$, i.e. $A_{p-1}^{j}=i$.
In the realization shown, the ancestral indices are $A_{0}=(1,2,4)$,
$A_{1}=(2,1,2)$ and $A_{2}=(3,2,2,3)$, while the Eve indices are
$E_{0}=(1,2,3,4)$, $E_{1}=(1,2,4)$, $E_{2}=(2,1,2)$ and $E_{3}=(2,1,1,2)$.}
\end{figure}

The particle approximations of $\eta_{n}$ and $\gamma_{n}$ are defined
respectively, with the convention $\prod_{p=0}^{-1}\eta_{p}^{N}(G_{p})=1$,
by the random measures
\[
\eta_{n}^{N}=\frac{1}{N_{n}}\sum_{i\in[N_{n}]}\delta_{\zeta_{n}^{i}},\qquad\gamma_{n}^{N}=\left\{ \prod_{p=0}^{n-1}\eta_{p}^{N}(G_{p})\right\} \eta_{n}^{N},
\]
and we observe that, similar to (\ref{eq:etan}), $\eta_{n}^{N}=\gamma_{n}^{N}/\gamma_{n}^{N}(1)$.
To simplify presentation, the dependence of $\gamma_{n}^{N}$ and
$\eta_{n}^{N}$ on $c_{0:n}$ is suppressed from the notation. The
following proposition establishes basic properties of the particle
approximations, which validate their use.

\begin{proposition}\label{prop:asconv-1}There exists a map $\sigma_{n}^{2}:\mathcal{L}(\mathcal{X})\rightarrow[0,\infty)$
such that for any $\varphi\in\mathcal{L}(\mathcal{X})$:
\begin{enumerate}
\item $E\left\{ \gamma_{n}^{N}(\varphi)\right\} =\gamma_{n}(\varphi)$,
for all $N\geq1$,
\item $\gamma_{n}^{N}(\varphi)\to\gamma_{n}(\varphi)$ almost surely and
$N{\rm var}\left\{ \gamma_{n}^{N}(\varphi)/\gamma_{n}(1)\right\} \to\sigma_{n}^{2}(\varphi)$,
\item $\eta_{n}^{N}(\varphi)\to\eta_{n}(\varphi)$ almost surely and $NE\left[\left\{ \eta_{n}^{N}(\varphi)-\eta_{n}(\varphi)\right\} ^{2}\right]\to\sigma_{n}^{2}(\varphi-\eta_{n}(\varphi))$.
\end{enumerate}
\end{proposition}In the case that the number of particles is constant
over time, $N_{p}=N$, these properties are well known and can be
deduced, for example, from various results of \citet{DelMoral2004}.
The arguments used to treat the general $N_{p}=\left\lceil c_{p}N\right\rceil $
case are not substantially different, but since they seem not to have
been published anywhere in exactly the form we need, we include a
proof of Proposition \ref{prop:asconv-1} in the supplement.

\subsection{A variance estimator}

For $\varphi\in\mathcal{L}(\mathcal{X})$, consider the quantity
\begin{equation}
V_{n}^{N}(\varphi)=\eta_{n}^{N}(\varphi)^{2}-\left(\prod_{p=0}^{n-1}\frac{N_{p}}{N_{p}-1}\right)\frac{1}{N_{n}(N_{n}-1)}\sum_{i,j:E_{n}^{i}\neq E_{n}^{j}}\varphi(\zeta_{n}^{i})\varphi(\zeta_{n}^{j}),\label{eq:V_n^N_defn_front}
\end{equation}
which is readily computable as a byproduct of Algorithm \ref{alg:bpf}.
The following theorem is the first main result of the paper. We state
it here to make some of the practical implications of our work accessible
to the reader before entering into more technical details; it shows
that via (\ref{eq:V_n^N_defn_front}), the Eve variables $E_{n}^{i}$
can be used to estimate the Monte Carlo errors associated with $\gamma_{n}^{N}(\varphi)$
and $\eta_{n}^{N}(\varphi)$. 

\begin{theorem}\label{thm:V_n^N_thm_front}The following hold for
any $\varphi\in\mathcal{L}(\mathcal{X})$, with $\sigma_{n}^{2}(\cdot)$
as in Proposition~\ref{prop:asconv-1}:
\begin{enumerate}
\item $E\left\{ \gamma_{n}^{N}(1)^{2}V_{n}^{N}(\varphi)\right\} ={\rm var}\left\{ \gamma_{n}^{N}(\varphi)\right\} $
for all $N\geq1$, 
\item $NV_{n}^{N}(\varphi)\to\sigma_{n}^{2}(\varphi)$ in probability, 
\item $NV_{n}^{N}(\varphi-\eta_{n}^{N}(\varphi))\to\sigma_{n}^{2}(\varphi-\eta_{n}(\varphi))$
in probability.
\end{enumerate}
\end{theorem}

The proof of Theorem~\ref{thm:V_n^N_thm_front}, given in the appendix,
relies on a number of intermediate results concerning moment properties
of the particle approximations which we shall develop in the coming
sections. Before embarking on this development let us discuss how
(\ref{eq:V_n^N_defn_front}) may be interpreted. Consider random variables
$X^{1},\ldots,X^{N}$ with sample mean $\bar{X}$ and sample variance
\begin{equation}
\bar{X}^{2}-\frac{1}{N(N-1)}\sum_{i\neq j}X^{i}X^{j}=\frac{1}{N(N-1)}\sum_{i}(X^{i}-\bar{X})^{2}.\label{eq:sample_var}
\end{equation}
When $X^{1},\ldots,X^{N}$ are independent and identically distributed
it is of course elementary that (\ref{eq:sample_var}) is an unbiased
estimator of $\mathrm{var}(\bar{X})$ and consistency properties are
easily established. Observe the resemblance between (\ref{eq:V_n^N_defn_front})
and the left hand side of (\ref{eq:sample_var}). Some of the features
which distinguish these two expressions, notably the summation over
$\{(i,j):E_{n}^{i}\neq E_{n}^{j}\}$ and the product term in (\ref{eq:V_n^N_defn_front}),
are a reflection of the dependence between the particles and specific
distributional characteristics of Algorithm~\ref{alg:bpf}. One of
the main difficulties we face is to develop a suitable mathematical
perspective from which to describe this dependence and thus establish
that (\ref{eq:V_n^N_defn_front}) does indeed have the properties
stated in Theorem \ref{thm:V_n^N_thm_front}.

It seems natural to ask if (\ref{eq:V_n^N_defn_front}) can be re-written
so as to resemble the right hand side of (\ref{eq:sample_var}) and
thus be interpreted as some kind of sample variance across the population
of particles. This motivates the following corollary, using the notation
\[
\#_{n}^{i}=\mathrm{card}\{j:E_{n}^{j}=i\},\quad\Delta_{n}^{i}=\frac{1}{\#_{n}^{i}}\sum_{j:E_{n}^{j}=i}\varphi(\zeta_{n}^{j})-\eta_{n}^{N}(\varphi),
\]
with the convention $\Delta_{n}^{i}=0$ when $\#_{n}^{i}=0$. Recall
from Section~\ref{sub:Statistical-examples-of} that in the hidden
Markov model and sequential simulation examples $\gamma_{n}(1)$ is
respectively the marginal likelihood and ratio of normalizing constants,
hence our interest in $V_{n}^{N}(\varphi)$ with specifically $\varphi=1$.

\begin{corollary}\label{cor:sample_variances}In the case that $c_{p}=1$
for all $p\in\{0,\ldots,n\}$, 
\begin{eqnarray}
 &  & NV_{n}^{N}(1)=\frac{1}{N}\sum_{i\in[N]}(\#_{n}^{i}-1)^{2}-n+\mathcal{O}_{p}(1/N),\label{eq:NV(1)_sample_var}\\
 &  & NV_{n}^{N}(\varphi-\eta_{n}^{N}(\varphi))=\frac{1}{N}\sum_{i\in[N]}(\#_{n}^{i}\Delta_{n}^{i})^{2}+\mathcal{O}_{p}(1/N).\label{eq:NV(phi-etaphi)_sample_var}
\end{eqnarray}
\end{corollary}The proof is in the supplement.

Since $\sum_{i}\#_{n}^{i}=N$, the first term on the right hand side
of (\ref{eq:NV(1)_sample_var}) can be interpreted as a sample variance
of the $\#_{n}^{i}$'s, reflecting variation in the numbers of time
$n$ descendants across the population of time $0$ particles. Since
$\sum_{i}\#_{n}^{i}\Delta_{n}^{i}=0$, the first term on the right
hand side of (\ref{eq:NV(phi-etaphi)_sample_var}) can be interpreted
as a sample variance which reflects both variation in the $\#_{n}^{i}$'s
and the deviations of the familial means $(\#_{n}^{i})^{-1}\sum_{j:E_{n}^{j}=i}\varphi(\zeta_{n}^{j})$
from the population mean $\eta_{n}^{N}(\varphi)$. Note that when
$n=0$, $E_{0}^{i}=i$ always and
\[
V_{0}^{N}(\varphi-\eta_{0}^{N}(\varphi))=\frac{1}{N_{0}(N_{0}-1)}\sum_{i\in[N_{0}]}\left\{ \varphi(\zeta_{0}^{i})-\eta_{0}^{N}(\varphi)\right\} ^{2},
\]
which is in keeping with $\zeta_{0}^{i}$ being independent and identically
distributed according to $\eta_{0}$.

\section{Moment properties of the particle approximations\label{sec:Lack-of-bias-and-second}}

\subsection{Genealogical tracing variables\label{sub:Genealogical-tracing-variables}}

Our next step is to introduce some auxiliary random variables associated
with the genealogical structure of the particle system. These auxiliary
variables are introduced only for purposes of analysis: they will
assist in deriving and justifying our variance estimators. Given $(A,\zeta)$,
the first collection of variables, $K^{1}=(K_{0}^{1},\ldots,K_{n}^{1})$,
is conditionally distributed as follows: $K_{n}^{1}$ is uniformly
distributed on $[N_{n}]$ and for each $p=n-1,\ldots,0$, $K_{p}^{1}=A_{p}^{K_{p+1}^{1}}$.
Given $(A,\zeta)$ and $K^{1}$, the second collection of variables,
$K^{2}=(K_{0}^{2},\ldots,K_{n}^{2})$, is conditionally distributed
as follows: $K_{n}^{2}$ is uniformly distributed on $[N_{n}]$ and
for each $p=n-1,\ldots,0$ we have $K_{p}^{2}=A_{p}^{K_{p+1}^{2}}$
if $K_{p+1}^{2}\neq K_{p+1}^{1}$ and $K_{p}^{2}\sim\mathcal{C}(G_{p}(\zeta_{p}^{1}),\ldots,G_{p}(\zeta_{p}^{N_{p}}))$
if $K_{p+1}^{2}=K_{p+1}^{1}$. The interpretation of $K^{1}$ is that
it traces backwards in time the ancestral lineage of a particle chosen
randomly from the population at time $n$. $K^{2}$ is slightly more
complicated: it traces backwards in time a sequence of broken ancestral
lineages, where breaks in the lineages occur when components of $K^{1}$
and $K^{2}$ coincide.

\subsection{Lack-of-bias and second moment of $\gamma_{n}^{N}(\varphi)$\label{sub:Lack-of-bias-and-second}}

We now give expressions for the first two moments of $\gamma_{n}^{N}(\varphi)$.

\begin{lemma}\label{LEM:Q1_CHANGE}For any $\varphi\in\mathcal{L}(\mathcal{X})$,
$E\left\{ \gamma_{n}^{N}(1)\varphi(\zeta_{n}^{K_{n}^{1}})\right\} =\gamma_{n}(\varphi)$
and $E\left\{ \gamma_{n}^{N}(\varphi)\right\} =\gamma_{n}(\varphi)$.\end{lemma}The
proof is in the supplement. This lack-of-bias property $E\left\{ \gamma_{n}^{N}(\varphi)\right\} =\gamma_{n}(\varphi)$
is quite well known and a martingale proof for the $N_{p}=N$ case
can be found, for example, in \citet[Ch. 9]{DelMoral2004}. 

In order to present an expression for the second moment of $\gamma_{n}^{N}(\varphi)$,
we now introduce a collection of measures on $\mathcal{X}^{\otimes2}$,
denoted $\{\mu_{b}:b\in B_{n}\}$ where $B_{n}=\{0,1\}^{n+1}$ is
the set of binary strings of length $n+1$. The measures are constructed
as follows. For a given $b\in B_{n}$, let $(X_{p},X_{p}')_{0\leq p\leq n}$
be a Markov chain with state-space $\mathsf{X}^{2}$, distributed
according to the following recipe. If $b_{0}=0$ then $X_{0}\sim M_{0}$
and $X_{0}'\sim M_{0}$ independently, while if $b_{0}=1$ then $X_{0}'=X_{0}\sim M_{0}$.
Then, for $p=1,\ldots,n$, if $b_{p}=0$ then $X_{p}\sim M_{p}(X_{p-1},\cdot)$
and $X_{p}'\sim M_{p}(X_{p-1}',\cdot)$ independently, while if $b_{p}=1$
then $X_{p}'=X_{p}\sim M_{p}(X_{p-1},\cdot)$. Letting $E_{b}$ denote
expectation with respect to the law of this Markov chain we then define
\[
\mu_{b}(S)=E_{b}\left[\mathbb{I}\left\{ (X_{n},X_{n}')\in S\right\} \prod_{p=0}^{n-1}G_{p}(X_{p})G_{p}(X_{p}')\right],\qquad S\in\mathcal{X}^{\otimes2},\quad b\in B_{n}.
\]
Similarly to (\ref{eq:gammanphi_expectation}) we shall write $\mu_{b}(\varphi)=E_{b}\left\{ \varphi(X_{n},X_{n}')\prod_{p=0}^{n-1}G_{p}(X_{p})G_{p}(X_{p}')\right\} $,
for $\varphi\in\mathcal{L}(\mathcal{X}^{\otimes2})$ and $b\in B_{n}$.

\begin{remark}\label{rem:mu0}Observe that with $0_{n}\in B_{n}$
denoting the zero string, $\mu_{0_{n}}(\varphi^{\otimes2})=\gamma_{n}(\varphi)^{2}$.\end{remark}

Let $[N_{0:n}]=[N_{0}]\times\cdots\times[N_{n}]$, and for any $b\in B_{n}$,
\[
\mathcal{I}(b)=\{(k^{1},k^{2})\in[N_{0:n}]^{2}\,:\,\text{for each }p,\:k_{p}^{1}=k_{p}^{2}\iff b_{p}=1\},
\]
which is the set of pairs of $[N_{0:n}]$-valued strings which coincide
in their $p$-th coordinate exactly when $b_{p}=1$.

\begin{lemma}\label{lem:Q2_change}For any $\varphi\in\mathcal{L}(\mathcal{X}^{\otimes2})$
and $b\in B_{n}$,
\begin{equation}
E\left[\mathbb{I}\left\{ (K^{1},K^{2})\in\mathcal{I}(b)\right\} \gamma_{n}^{N}(1)^{2}\varphi(\zeta_{n}^{K_{n}^{1}},\zeta_{n}^{K_{n}^{2}})\right]=\prod_{p=0}^{n}\left\{ \left(\frac{1}{N_{p}}\right)^{b_{p}}\left(1-\frac{1}{N_{p}}\right)^{1-b_{p}}\right\} \mu_{b}(\varphi)\label{eq:second_moment_formula_disintegrated}
\end{equation}
and
\begin{equation}
E\left\{ \gamma_{n}^{N}(\varphi)^{2}\right\} =\sum_{b\in B_{n}}\prod_{p=0}^{n}\left(\frac{1}{N_{p}}\right)^{b_{p}}\left(1-\frac{1}{N_{p}}\right)^{1-b_{p}}\mu_{b}(\varphi^{\otimes2}).\label{eq:second_moment_formula}
\end{equation}
\end{lemma}The proof of Lemma~\ref{lem:Q2_change} is in the supplement
and uses an argument involving the law of a doubly conditional sequential
Monte Carlo algorithm \citep[see also][]{andrieu2013uniform}. The
identity (\ref{eq:second_moment_formula}) was first proved by \citet{cerou2011nonasymptotic}
in the case where $N_{p}=N$. Our proof technique is different: we
obtain (\ref{eq:second_moment_formula}) as a consequence of (\ref{eq:second_moment_formula_disintegrated}).
The appearance of $K^{1},K^{2}$ in (\ref{eq:second_moment_formula_disintegrated})
is also central to the justification of our variance estimators below.

\subsection{Asymptotic variances}

For each $p\in\{0,\ldots,n\}$, we denote by $e_{p}\in B_{n}$ the
vector with a $1$ in position $p$ and zeros elsewhere. As in Remark~\ref{rem:mu0},
$0_{n}$ denotes the zero string in $B_{n}$. The following result
builds upon Lemmas~\ref{LEM:Q1_CHANGE}--\ref{lem:Q2_change}. It
shows that a particular subset of the measures $\{\mu_{b}:b\in B_{n}\}$,
namely $\mu_{0_{n}}$ and $\{\mu_{e_{p}}:p=0,\ldots,n\}$, appear
in the asymptotic variances.

\begin{lemma}\label{lem:asymptotic_variances}For any $\varphi\in\mathcal{L}(\mathcal{X})$,
define 
\begin{equation}
v_{p,n}(\varphi)=\frac{\mu_{e_{p}}(\varphi^{\otimes2})-\mu_{0_{n}}(\varphi^{\otimes2})}{\gamma_{n}(1)^{2}},\qquad p\in\{0,\ldots,n\}.\label{eq:vpndefn}
\end{equation}
Then $N\mathrm{var}\left\{ \gamma_{n}^{N}(\varphi)/\gamma_{n}(1)\right\} \to\sum_{p=0}^{n}c_{p}^{-1}v_{p,n}(\varphi)$
and 
\begin{equation}
NE\left[\left\{ \eta_{n}^{N}(\varphi)-\eta_{n}(\varphi)\right\} ^{2}\right]\to\sum_{p=0}^{n}c_{p}^{-1}v_{p,n}(\varphi-\eta_{n}(\varphi)).\label{eq:mseetan}
\end{equation}
 \end{lemma}The proof of Lemma~\ref{lem:asymptotic_variances} is
in the supplement. 

\begin{remark}\label{rem:vpn_Q_expr}In light of Lemma \ref{lem:asymptotic_variances},
the map $\sigma_{n}^{2}$ in Proposition~\ref{prop:asconv-1} satisfies
\begin{equation}
\sigma_{n}^{2}(\varphi)=\sum_{p=0}^{n}c_{p}^{-1}v_{p,n}(\varphi),\qquad\varphi\in\mathcal{L}(\mathcal{X}).\label{eq:sigma_equals_sum_v}
\end{equation}
An expression for $v_{p,n}(\varphi)$ in terms of $(M_{p},G_{p})_{0\leq p\leq n}$
is obtained by observing that if we define 
\[
Q_{p}(x_{p-1},{\rm d}x_{p})=G_{p-1}(x_{p-1})M_{p}(x_{p-1},{\rm d}x_{p}),\quad p\in\{1,\ldots,n\},
\]
and $Q_{n,n}=Id$, $Q_{p,n}=Q_{p+1}\cdots Q_{n}$ for $p\in\{0,\ldots,n-1\}$,
then $\mu_{e_{p}}(\varphi)=\gamma_{p}(Q_{p,n}(\varphi)^{2})$. In
combination with Remark~\ref{rem:mu0}, we obtain
\begin{equation}
v_{p,n}(\varphi)=\frac{\gamma_{p}(1)\gamma_{p}(Q_{p,n}(\varphi)^{2})}{\gamma_{n}(1)^{2}}-\eta_{n}(\varphi)^{2}=\frac{\eta_{p}(Q_{p,n}(\varphi)^{2})}{\eta_{p}Q_{p,n}(1)^{2}}-\eta_{n}(\varphi)^{2}.\label{eq:vpn_Q_expr}
\end{equation}

\end{remark}

\section{The estimators\label{SEC:COMPUTATION_OF_ESTS}}

\subsection{Particle approximations of each $\mu_{b}$\label{SEC:PARTICLE_APPROXIMATIONS}}

We now introduce particle approximations of the measures $\{\mu_{b}:b\in B_{n}\}$,
from which we shall subsequently derive the variance estimators. For
each $b\in B_{n}$, and $\varphi\in\mathcal{L}(\mathcal{X}^{\otimes2})$
we define
\begin{equation}
\mu_{b}^{N}(\varphi)=\left[\prod_{p=0}^{n}\left(N_{p}\right)^{b_{p}}\left(\frac{N_{p}}{N_{p}-1}\right)^{1-b_{p}}\right]\gamma_{n}^{N}(1)^{2}E\left[\mathbb{I}\left\{ (K^{1},K^{2})\in\mathcal{I}(b)\right\} \varphi(\zeta_{n}^{K_{n}^{1}},\zeta_{n}^{K_{n}^{2}})\mid A,\zeta\right].\label{eq:mubN_defn}
\end{equation}
Recalling from Section \ref{sub:Genealogical-tracing-variables} that
given $A$ and $\zeta$, $K_{n}^{1}$ and $K_{n}^{2}$ are conditionally
independent and each uniformly distributed on $[N_{n}]$, it follows
from (\ref{eq:mubN_defn}) that 
\begin{eqnarray}
\gamma_{n}^{N}(\varphi)^{2} & = & \gamma_{n}^{N}(1)^{2}\frac{1}{N_{n}^{2}}\sum_{i,j\in[N_{n}]}\varphi(\zeta_{n}^{i})\varphi(\zeta_{n}^{j})\nonumber \\
 & = & \gamma_{n}^{N}(1)^{2}\sum_{b\in B_{n}}E\left[\mathbb{I}\left\{ (K^{1},K^{2})\in\mathcal{I}(b)\right\} \varphi(\zeta_{n}^{K_{n}^{1}})\varphi(\zeta_{n}^{K_{n}^{2}})\mid A,\zeta\right]\nonumber \\
 & = & \sum_{b\in B_{n}}\left\{ \prod_{p=0}^{n}\left(\frac{1}{N_{p}}\right)^{b_{p}}\left(1-\frac{1}{N_{p}}\right)^{1-b_{p}}\right\} \mu_{b}^{N}(\varphi^{\otimes2}),\label{eq:gamma_n^N_mu_identity}
\end{eqnarray}
mirroring (\ref{eq:second_moment_formula}). This identity is complemented
by the following result.

\begin{theorem}\label{THM:MU_ B^N}For any $b\in B_{n}$ and $\varphi\in\mathcal{L}(\mathcal{X}^{\otimes2})$,
\begin{enumerate}
\item $E\left\{ \mu_{b}^{N}(\varphi)\right\} =\mu_{b}(\varphi)$ for all
$N\geq1$,
\item $\sup_{N\geq1}NE\left[\left\{ \mu_{b}^{N}(\varphi)-\mu_{b}(\varphi)\right\} ^{2}\right]<\infty$
and hence $\mu_{b}^{N}(\varphi)\to\mu_{b}(\varphi)$ in probability. 
\end{enumerate}
\end{theorem}

The proof of Theorem~\ref{THM:MU_ B^N} is in the supplement. Although
(\ref{eq:mubN_defn}) can be computed in principle from the output
of Algorithm \ref{alg:bpf} without the need for any further simulation,
the conditional expectation in (\ref{eq:mubN_defn}) involves a summation
over all binary strings in $\mathcal{I}(b)$, so calculating $\mu_{b}^{N}(\varphi^{\otimes2})$
in practice may be computationally expensive. Fortunately, relatively
simple and computationally efficient expressions are available for
$\mu_{b}^{N}(\varphi^{\otimes2})$ in the cases $b=0_{n}$ and $b=e_{p}$,
and those are the only ones required to construct our variance estimators.

\subsection{Variance estimators\label{sub:Variance-estimators}}

Our next objective is to explain how (\ref{eq:V_n^N_defn_front})
is related to the measures $\mu_{b}^{N}$ and to introduce another
family of estimators associated with the individual terms in (\ref{eq:sigma_equals_sum_v}).
We need the following technical lemma.

\begin{lemma}\label{lem:eve_identical_events}The following identity
of events holds: $\left\{ E_{n}^{K_{n}^{1}}\neq E_{n}^{K_{n}^{2}}\right\} =\left\{ (K^{1},K^{2})\in\mathcal{I}(0_{n})\right\} $.\end{lemma}

The proof is in the appendix. Combined with the fact that given $(A,\zeta)$,
$K_{n}^{1},K_{n}^{2}$ are independent and identically distributed
according to the uniform distribution on $[N_{n}]$, we have
\begin{equation}
E\left[\mathbb{I}\left\{ (K^{1},K^{2})\in\mathcal{I}(0_{n})\right\} \varphi(\zeta_{n}^{K_{n}^{1}},\zeta_{n}^{K_{n}^{2}})\mid A,\zeta\right]=N_{n}^{-2}\sum_{i,j:E_{n}^{i}\neq E_{n}^{j}}\varphi(\zeta_{n}^{i})\varphi(\zeta_{n}^{j}),\label{eq:eveeqn}
\end{equation}
and therefore we arrive at the following equivalent of (\ref{eq:V_n^N_defn_front}),
written in terms of $\mu_{0_{n}}^{N}$,
\begin{eqnarray}
V_{n}^{N}(\varphi) & = & \eta_{n}^{N}(\varphi)^{2}-\frac{\mu_{0_{n}}^{N}(\varphi^{\otimes2})}{\gamma_{n}^{N}(1)^{2}}.\label{eq:VnNdefn}
\end{eqnarray}
Detailed pseudocode for computing $V_{n}^{N}(\varphi)$ in $\mathcal{O}(N)$
time and space upon running Algorithm~\ref{alg:bpf} is provided
in the supplement.

Mirroring (\ref{eq:vpndefn}), we now define
\[
v_{p,n}^{N}(\varphi)=\frac{\mu_{e_{p}}^{N}(\varphi^{\otimes2})-\mu_{0_{n}}^{N}(\varphi^{\otimes2})}{\gamma_{n}^{N}(1)^{2}},\quad p\in\{0,\ldots,n\},\quad\quad v_{n}^{N}(\varphi)=\sum_{p=0}^{n}c_{p}^{-1}v_{p,n}^{N}(\varphi).
\]
Detailed pseudocode for computing each $v_{p,n}^{N}(\varphi)$ and
$v_{n}^{N}(\varphi)$ with time and space complexity in $\mathcal{O}(Nn)$
time upon running Algorithm~\ref{alg:bpf} is provided in the supplement.
The time complexity is the same as that of running Algorithm~\ref{alg:bpf},
but the space complexity is larger. Empirically, we have found that
$NV_{n}^{N}(\varphi)$ is very similar to $v_{n}^{N}(\varphi)$ as
an estimator of $\sigma_{n}^{2}(\varphi)$ when $N$ is large enough
that they are both accurate, and hence may be preferable due to its
reduced space complexity.

\begin{theorem}\label{thm:v_pn^N}For any $\varphi\in\mathcal{L}(\mathcal{X})$,
\begin{enumerate}
\item $E\left\{ \gamma_{n}^{N}(1)^{2}v_{p,n}^{N}(\varphi)\right\} =\gamma_{n}(1)^{2}v_{p,n}(\varphi)$
for all $N\geq1$,
\item $v_{p,n}^{N}(\varphi)\to v_{p,n}(\varphi)$ and $v_{p,n}^{N}(\varphi-\eta_{n}^{N}(\varphi))\to v_{p,n}(\varphi-\eta_{n}(\varphi))$,
both in probability, 
\item $E\left\{ \gamma_{n}^{N}(1)^{2}v_{n}^{N}(\varphi)\right\} =\gamma_{n}(1)^{2}\sigma_{n}^{2}(\varphi)$
for all $N\geq1$ and $v_{n}^{N}(\varphi)\to\sigma_{n}^{2}(\varphi)$
in probability.
\end{enumerate}
\end{theorem}

\section{Estimators for updated measures\label{SEC:FILTERING_TRANSLATION}}

In some applications there is interest in approximating the updated
measures: 
\[
\hat{\gamma}_{n}(S)=\int_{S}G_{n}(x)\gamma_{n}({\rm d}x),\qquad\hat{\eta}_{n}(S)=\frac{\hat{\gamma}_{n}(S)}{\hat{\gamma}_{n}(1)},\qquad S\in\mathcal{X}.
\]
In the hidden Markov model setting described in Section~\ref{sub:Discrete-time-Feynman--Kac},
e.g., $\hat{\eta}_{n}$ is the conditional distribution of $X_{n}$
given $y_{0},\ldots,y_{n}$, that is $\hat{\eta}_{n}$ is a filtering
distribution, while $\eta_{n}$ is a predictive distribution.

The updated particle approximations are defined by
\[
\hat{\gamma}_{n}^{N}(S)=\int_{S}G_{n}(x)\gamma_{n}^{N}({\rm d}x),\qquad\hat{\eta}_{n}^{N}(S)=\frac{\hat{\gamma}_{n}^{N}(S)}{\hat{\gamma}_{n}^{N}(1)},\qquad S\in\mathcal{X},
\]
and we now define their variance estimators. To facilitate this task,
we consider a fixed $\varphi\in\mathcal{L}(\mathcal{X})$, and define
$\hat{\varphi}(x)=G_{n}(x)\varphi(x)$. The following relationships
can then be deduced: $\hat{\gamma}_{n}(\varphi)\equiv\gamma_{n}(\hat{\varphi})$,
$\hat{\eta}_{n}(\varphi)\equiv\eta_{n}(\hat{\varphi})/\eta_{n}(G_{n})$,
$\hat{\gamma}_{n}^{N}(\varphi)\equiv\gamma_{n}^{N}(\hat{\varphi})$
and $\hat{\eta}_{n}^{N}(\varphi)\equiv\eta_{n}^{N}(\hat{\varphi})/\eta_{n}^{N}(G_{n})$.
We define analogues of $\sigma_{n}^{2}$ and $v_{p,n}$ for the updated
particle approximations as 
\[
\hat{\sigma}_{n}^{2}(\varphi)=\lim_{N\to\infty}N\text{var}\left\{ \hat{\gamma}_{n}^{N}(\varphi)/\hat{\gamma}_{n}(1)\right\} ,\qquad\hat{v}_{p,n}(\varphi)=\frac{v_{p,n}(\hat{\varphi})}{\eta_{n}(G_{n})^{2}},
\]
and the proposition below is a counterpart to Proposition~\ref{prop:asconv-1}
and Lemma~\ref{lem:asymptotic_variances}.

\begin{proposition}\label{prop:avar_hats}For any $\varphi\in\mathcal{L}(\mathcal{X})$,
\begin{enumerate}
\item $\hat{\gamma}_{n}^{N}(\varphi)\to\hat{\gamma}_{n}(\varphi)$ almost
surely and $\hat{\sigma}_{n}^{2}(\varphi)={\displaystyle \sum_{p=0}^{n}c_{p}^{-1}\hat{v}_{p,n}(\varphi)}$,
\item $\hat{\eta}_{n}^{N}(\varphi)\to\hat{\eta}_{n}(\varphi)$ almost surely
and $NE\left[\left\{ \hat{\eta}_{n}^{N}(\varphi)-\hat{\eta}_{n}(\varphi)\right\} ^{2}\right]\to\hat{\sigma}_{n}^{2}(\varphi-\hat{\eta}_{n}(\varphi))$.
\end{enumerate}
\end{proposition}

The proofs of Proposition~\ref{prop:avar_hats}, and Theorems~\ref{thm:VhatnN}--\ref{thm:vhatpnN}
below can be found in the supplement. Proposition~\ref{prop:avar_hats}
implies the relationship $\hat{\sigma}_{n}^{2}(\varphi)=\sigma_{n}^{2}(\hat{\varphi})/\eta_{n}(G_{n})^{2}$.
The corresponding estimates of the variance, asymptotic variance and
the terms therein are now obtained and analogues of Theorems~\ref{thm:V_n^N_thm_front}
and~\ref{thm:v_pn^N} follow straightforwardly. Below we write the
estimators $\hat{V}_{n}^{N},$ $\hat{v}_{p,n}^{N}$ etc. in terms
of $V_{n}^{N}$, $\eta_{n}^{N}$ and $v_{p,n}^{N}$ to emphasize that
the same algorithms can be used to compute them, just as $\hat{\gamma}_{n}^{N}(\varphi)$
and $\hat{\eta}_{n}^{N}(\varphi)$ can be computed as $\gamma_{n}^{N}(\hat{\varphi})$
and $\eta_{n}^{N}(\hat{\varphi})/\eta_{n}^{N}(G_{n})$, respectively.

\begin{theorem}\label{thm:VhatnN}For any $\varphi\in\mathcal{L}(\mathcal{X})$,
with 
\begin{equation}
\hat{V}_{n}^{N}(\varphi)=V_{n}^{N}(\hat{\varphi})/\eta_{n}^{N}(G_{n})^{2},\label{eq:hat_VnN}
\end{equation}

\begin{enumerate}
\item $E\left\{ \hat{\gamma}_{n}^{N}(1)^{2}\hat{V}_{n}^{N}(\varphi)\right\} ={\rm var}\left\{ \hat{\gamma}_{n}^{N}(\varphi)\right\} $
for all $N\geq1$,
\item $N\hat{V}_{n}^{N}(\varphi)\to\hat{\sigma}_{n}^{2}(\varphi)$ in probability,
\item $N\hat{V}_{n}^{N}(\varphi-\hat{\eta}_{n}^{N}(\varphi))\to\hat{\sigma}_{n}^{2}(\varphi-\hat{\eta}_{n}(\varphi))$
in probability.
\end{enumerate}
\end{theorem}

\begin{remark}It follows from (\ref{eq:V_n^N_defn_front}), (\ref{eq:VnNdefn}),
(\ref{eq:hat_VnN}) and simple manipulations that 
\[
\frac{N\hat{V}_{n}^{N}(\varphi-\hat{\eta}_{n}^{N}(\varphi))}{\left(\prod_{p=0}^{n}\frac{N_{p}}{N_{p}-1}\right)}=N\sum_{i\in[N_{0}]}\left[\frac{\sum_{j\in[N_{n}]:E_{n}^{j}=i}G_{n}(\zeta_{n}^{j})\left\{ \varphi(\zeta_{n}^{j})-\hat{\eta}_{n}^{N}(\varphi)\right\} }{\sum_{j\in[N_{n}]}G_{n}(\zeta_{n}^{j})}\right]^{2},
\]
the right hand side of which is, in the case where $N$ is not time-varying,
precisely the estimator in Equation 2.9 of \citet{chan2013general}.\end{remark}

\begin{theorem}\label{thm:vhatpnN}For any $\varphi\in\mathcal{L}(\mathcal{X})$,
with 
\[
\hat{v}_{p,n}^{N}(\varphi)=v_{p,n}^{N}(\hat{\varphi})/\eta_{n}^{N}(G_{n})^{2},\qquad\hat{v}_{n}^{N}(\varphi)=\sum_{p=0}^{n}c_{p}^{-1}\hat{v}_{p,n}^{N}(\varphi),
\]

\begin{enumerate}
\item $E\left\{ \hat{\gamma}_{n}^{N}(1)^{2}\hat{v}_{p,n}^{N}(\varphi)\right\} =\hat{\gamma}_{n}(1)^{2}\hat{v}_{p,n}(\varphi)$
for all $N\geq1$,
\item $\hat{v}_{p,n}^{N}(\varphi)\to\hat{v}_{p,n}(\varphi)$ and $\hat{v}_{p,n}^{N}(\varphi-\hat{\eta}_{n}^{N}(\varphi))\to\hat{v}_{p,n}(\varphi-\hat{\eta}_{n}(\varphi))$,
both in probability,
\item $E\left\{ \hat{\gamma}_{n}^{N}(1)^{2}\hat{v}_{n}^{N}(\varphi)\right\} =\hat{\gamma}_{n}(1)^{2}\hat{\sigma}_{n}^{2}(\varphi)$
for all $N\geq1$ and $\hat{v}_{n}^{N}(\varphi)\to\hat{\sigma}_{n}^{2}(\varphi)$
in probability.
\end{enumerate}
\end{theorem}

\section{Use of the estimators to tune the particle filter\label{sec:Applications}}

The variance estimators we have proposed can of course be applied
directly to report estimates of Monte Carlo error alongside particle
approximations. Estimates of quantities such as $v_{p,n}(\varphi)$
may also aid algorithm and design. We provide here two simple examples
of adaptive methods to illustrate this, firstly concerning how to
improve performance by allowing particle numbers to vary over time,
and secondly concerning how to choose particle numbers so as to achieve
some user-defined performance criterion. To simplify presentation,
we focus on performance in estimating $\gamma_{n}^{N}(\varphi)$,
the ideas can easily be modified easily to deal instead with $\eta_{n}^{N}(\varphi)$,
$\hat{\gamma}_{n}^{N}(\varphi)$ or $\hat{\eta}_{n}^{N}(\varphi)$.

\subsection{Asymptotically optimal allocation\label{sub:approx_allocation}}

The following well known result is closely related to Neyman's optimal
allocation in stratified random sampling \citep{tschuprow1923mathematical,neyman1934two}.
A short proof using Jensen's inequality can be found in \citet[Section~4.3]{glasserman2003monte}.

\begin{lemma}Let $a_{0},\ldots,a_{n}\geq0$. The function $(c_{0},\ldots,c_{n})\mapsto\sum_{p=0}^{n}c_{p}^{-1}a_{p}$
is minimized, subject to the constraints $\min_{p}c_{p}>0$ and $\sum_{p=0}^{n}c_{p}=n+1$,
at $(n+1)^{-1}\left\Vert a\right\Vert _{2}^{2}$ when $c_{p}\propto a_{p}^{1/2}$.
\end{lemma}

As a consequence, we can in principle minimize $\sigma_{n}^{2}(\varphi)$
by choosing $c_{p}\propto v_{p,n}(\varphi)^{1/2}$. An approximation
of this optimal allocation can be obtained by the following two-stage
procedure. First run a particle filter with $N_{p}=N$ to obtain the
estimates $v_{p,n}^{N}(\varphi)$ and then define $c_{0:n}$ by $c_{p}=\max\left\{ v_{p,n}^{N}(\varphi),g(N)\right\} $,
where $g$ is some positive but decreasing function with $\lim_{N\rightarrow\infty}g(N)=0$.
Then run a second particle filter with each $N_{p}=\left\lceil c_{p}N\right\rceil $,
and report the quantities of interest, e.g., $\gamma_{n}^{N}(\varphi)$.
The function $g$ is chosen to ensure that $c_{p}>0$ and that for
large $N$ we permit small values of $c_{p}$. The quantity $\sum_{p=0}^{n}v_{p,n}^{N}(\varphi)/(\sum_{p=0}^{n}c_{p}^{-1}v_{p,n}^{N}(\varphi))$,
obtained from the first run, is an indication of the improvement in
variance obtained by using the new allocation.

Approximately optimal allocation has previously been addressed by
\citet{bhadra2014adaptive}, who introduced a meta-model to approximate
the distribution of the Monte Carlo error associated with $\log\gamma_{n}^{N}(1)$
in terms of an autoregressive process, the objective function to be
minimized then being the variance under this meta-model. They provide
only empirical evidence for the fit of their meta-model, whereas our
approach targets the true asymptotic variance $\sigma_{n}^{2}(\varphi)$
directly.

\subsection{An adaptive particle filter\label{sub:adaptiveN}}

Monte Carlo errors of particle filter approximations can be sensitive
to $N$, and an adequate value of $N$ to achieve a given error may
not be known \emph{a priori}. The following procedure increases $N$
until $V_{n}^{N}(\varphi)$ is in a given interval. 

Consider the case where we wish to estimate $\gamma_{n}(\varphi)$.
Given an initial number of particles $N^{(0)}$ and a threshold $\delta>0$,
one can run successive particle filters, doubling the number of particles
each time, until the associated random variable $V_{n}^{N^{(\tau)}}(\varphi)\in[0,\delta]$.
Finally, one runs a final particle filter with $N^{(\tau)}$ particles,
and returns the estimate of interest. We provide empirical evidence
in Section~\ref{sec:Numerical-illustrations} that this procedure
can be effective in some applications.

\section{Applications and illustrations\label{sec:Numerical-illustrations}}

In this section we demonstrate the empirical performance of the estimators
we have proposed in three examples. Our numerical results mostly address
the accuracy of our estimators of the asymptotic variance $\sigma_{n}^{2}(\varphi)$,
the individual terms $v_{p,n}(\varphi)$, and the effectiveness of
the applications described in Section~\ref{sec:Applications} with
the test functions $\varphi\equiv1$ and $\varphi=Id$, the identity
function. Where the results for later examples are qualitatively similar
to those of the first, the corresponding figures can be found in the
supplement.

\subsection{Linear Gaussian hidden Markov model\label{sub:Linear-Gaussian-HMM}}

\begin{figure}
\centering{}\subfloat[$\varphi\equiv1$]{\begin{centering}
\includegraphics[scale=0.5]{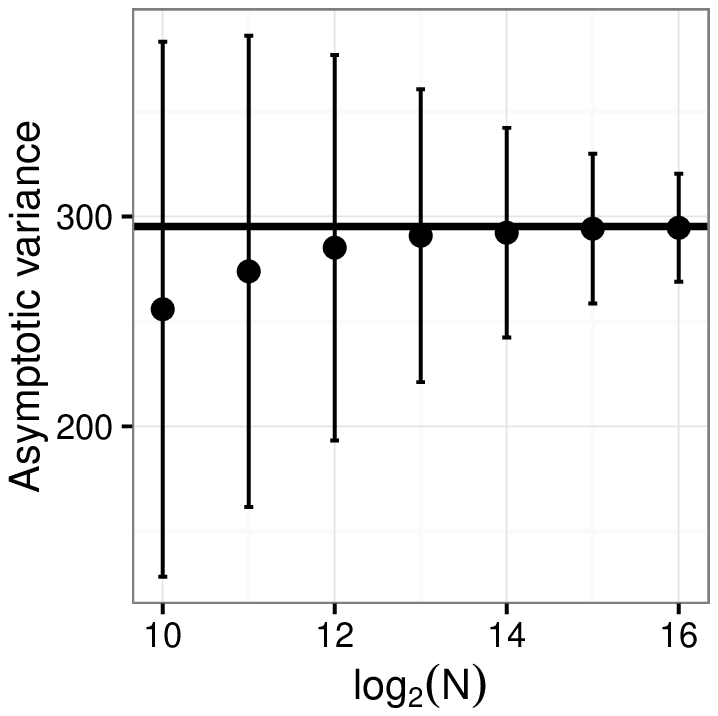}
\par\end{centering}

}\subfloat[$\varphi=Id-\hat{\eta}_{n}^{N}(Id)$]{\centering{}\includegraphics[scale=0.5]{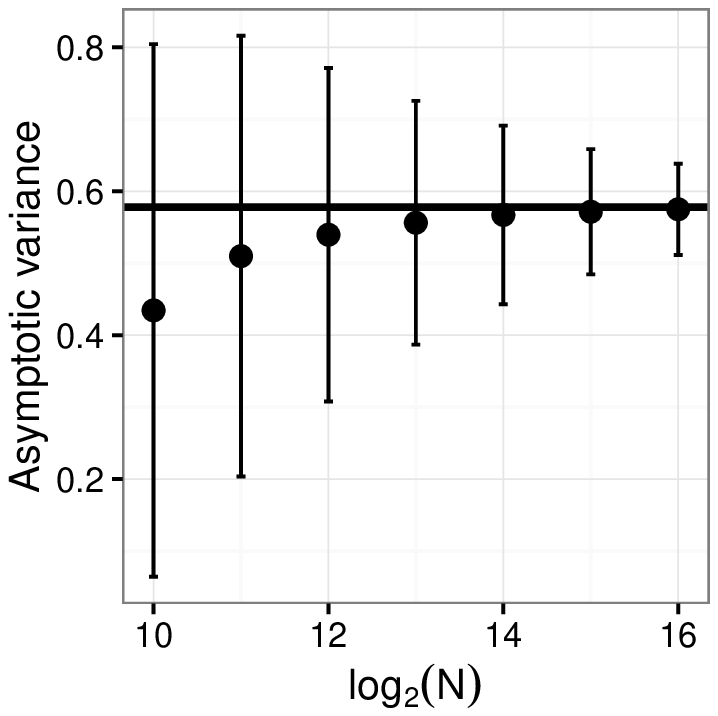}}\caption{Estimated asymptotic variances $N\hat{V}_{n}^{N}(\varphi)$ (dots
and error bars for the mean $\pm$ one standard deviation from $10^{4}$
replicates) against $\log_{2}N$ for the linear Gaussian example.
The horizontal lines correspond to the true asymptotic variances.
The sample variances of $\hat{\gamma}_{n}^{N}(1)/\hat{\gamma}_{n}(1)$
and $\hat{\eta}_{n}^{N}(Id)$, scaled by $N$, were close to their
asymptotic variances.\label{fig:lg2}}
\end{figure}

\begin{figure}
\centering{}\includegraphics[scale=0.6]{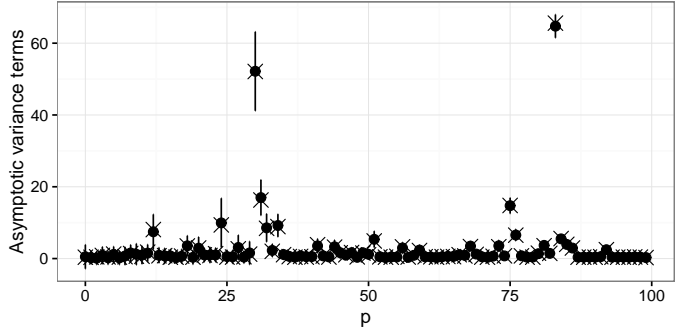}\caption{Plot of $\hat{v}_{p,n}^{N}(1)$ (dots and error bars for the mean
$\pm$ one standard deviation from $10^{4}$ replicates) and $\hat{v}_{p,n}(1)$
(crosses) at each $p\in\{0,\ldots,n\}$ for the Linear Gaussian example,
with $N=10^{5}$.\label{fig:lg1}}
\end{figure}

This model is specified by $M_{0}(\cdot)=\mathcal{N}(\cdot;0,1)$,
$M_{p}(x_{p-1},\cdot)=\mathcal{N}(\cdot;0.9x_{p-1},1)$ and $G_{p}(x_{p})=\mathcal{N}(y_{p};x_{p},1)$.
The measures $\hat{\eta}_{n}$ and $\hat{\gamma}_{n}$ are available
in closed form via the Kalman filter, and for suitable $\varphi$
the quantities $\hat{v}_{p,n}(\varphi)$ etc. can be computed exactly,
allowing us to assess the accuracy of our estimators. We used a synthetic
dataset, simulated according to the model with $n=99$. A Monte Carlo
study with $10^{4}$ replicates of $\hat{V}_{n}^{N}(\varphi)$ for
each value of $N$ and $c_{p}\equiv1$ was used to measure the accuracy
of the estimate $N\hat{V}_{n}^{N}(\varphi)$ as $N$ grows; results
are displayed in Figure~\ref{fig:lg2} and for this data $\hat{\sigma}_{n}^{2}(1)=295.206$
and $\hat{\sigma}_{n}^{2}(Id-\hat{\eta}_{n}(Id))\approx0.58$. The
estimates $\hat{v}_{n}^{N}(\varphi)$ differed very little from $N\hat{V}_{n}^{N}(\varphi)$,
and so are not shown. We then tested the accuracy of the estimates
$\hat{v}_{p,n}^{N}(1)$; results are displayed in Figure~\ref{fig:lg1}.
The estimates $\hat{v}_{p,n}^{N}(Id-\hat{\eta}_{n}^{N}(Id))$ are
very close to $0$ for $p<95$ and with values $(0.0017,0.012,0.082,0.48)$
for $p\in\{96,97,98,99\}$; this behaviour is in keeping with time-uniform
bounds on asymptotic variances obtained by \citet{whiteley2013stability},
see also references therein. 

We also compared a constant $N$ particle filter, the asymptotically
optimal particle filter where the asymptotically optimal allocation
is computed exactly, and its approximation described in Section~\ref{sub:approx_allocation}
for different values of $N_{}$ using a Monte Carlo study with $10^{4}$
replicates. We took $g(N)=2/\log_{2}N$ in defining the approximation,
and the results in Figure~\ref{fig:lg4a} indicate that indeed the
approximation reduces the variance. The improvement is fairly modest
for this particular model, and indeed the exact asymptotic variances
associated with the constant $N$ and asymptotically optimal particle
filters differ by less than a factor of $2$. In contrast, Figure~\ref{fig:lg4b}
shows that the improvement can be fairly dramatic in the presence
of outlying observations; the improvement in variance there is by
a factor of around $40$. Finally, we tested the adaptive particle
filter described in Section~\ref{sub:adaptiveN} using $10^{3}$
replicates for each value of $\delta$; results are displayed in Figure~\ref{fig:lg3},
and indicate that the estimates of $\hat{\gamma}_{n}(1)$ are close
to their prescribed thresholds.

\begin{figure}
\centering{}\subfloat[\label{fig:lg4a}]{\begin{centering}
\includegraphics[scale=0.5]{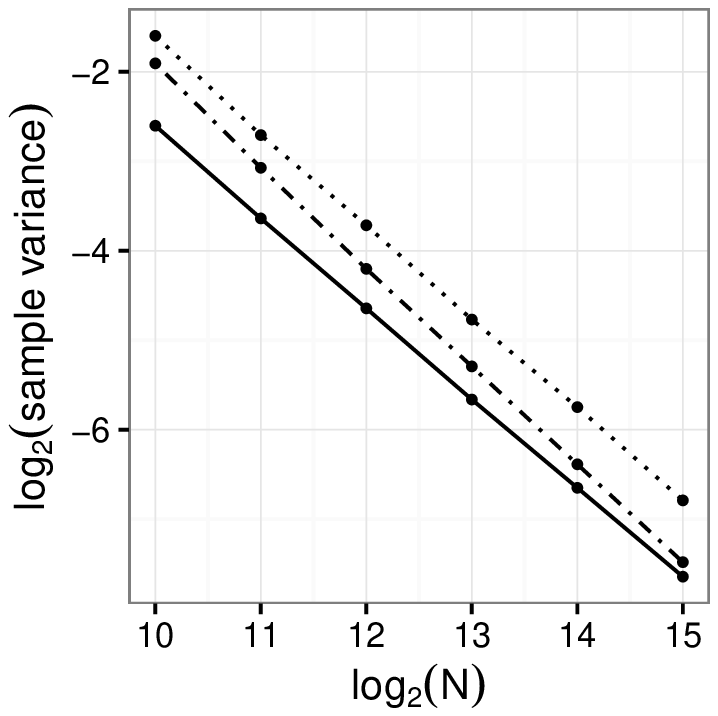}
\par\end{centering}

}\subfloat[\label{fig:lg4b}]{\centering{}\includegraphics[scale=0.5]{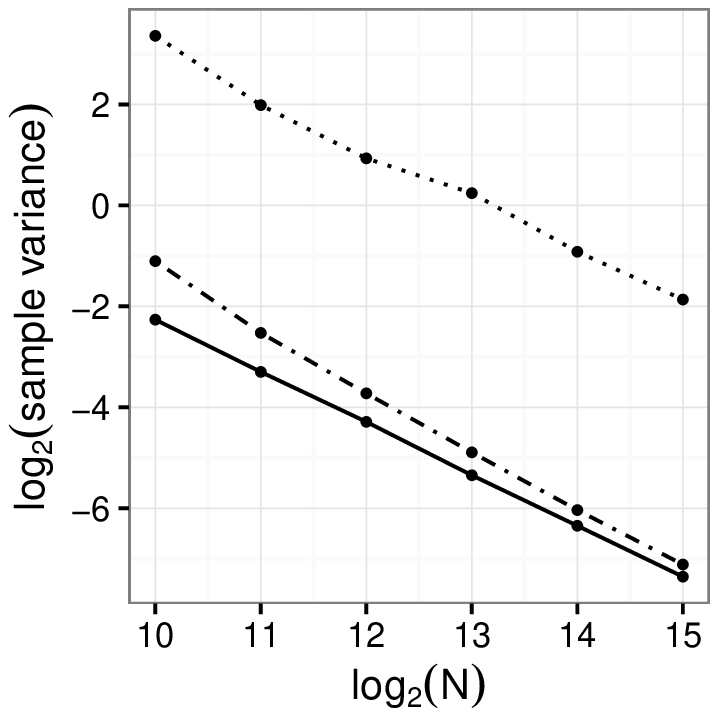}}\caption{Logarithmic plots of the sample variance across $10^{4}$ replicates
of $\gamma_{n}^{N}(1)/\gamma_{n}(1)$ against $N$ for the linear
Gaussian example, using a constant $N$ particle filter (dotted),
the approximation to the asymptotically optimal particle filter (dot-dash),
and the asymptotically optimal particle filter (solid). In Figure~\ref{fig:lg4b},
the observation sequence is $y_{p}=0$ for $p\in\{0,\ldots,99\}\setminus\{49\}$
and $y_{49}=8$.}
\end{figure}

\begin{figure}
\centering{}\subfloat[]{\begin{centering}
\includegraphics[scale=0.5]{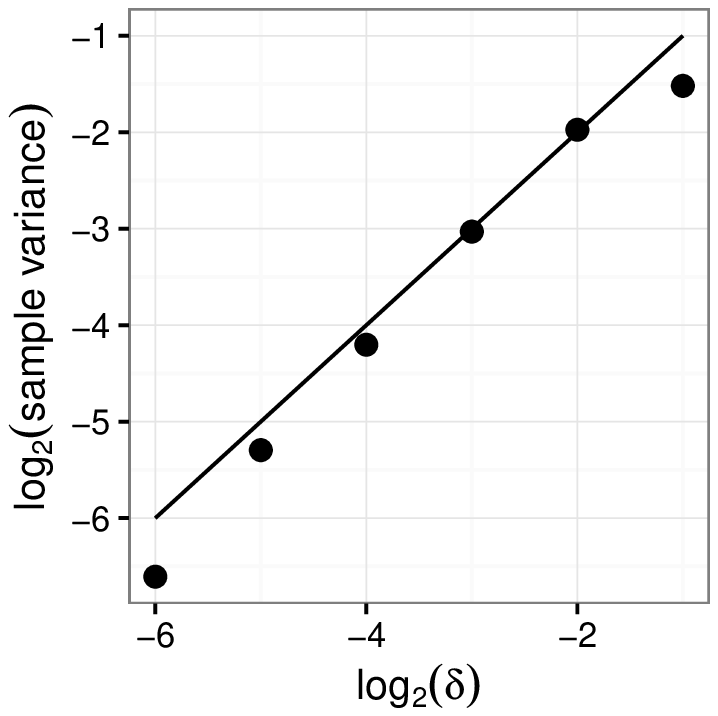}
\par\end{centering}

}\subfloat[]{\centering{}\includegraphics[scale=0.5]{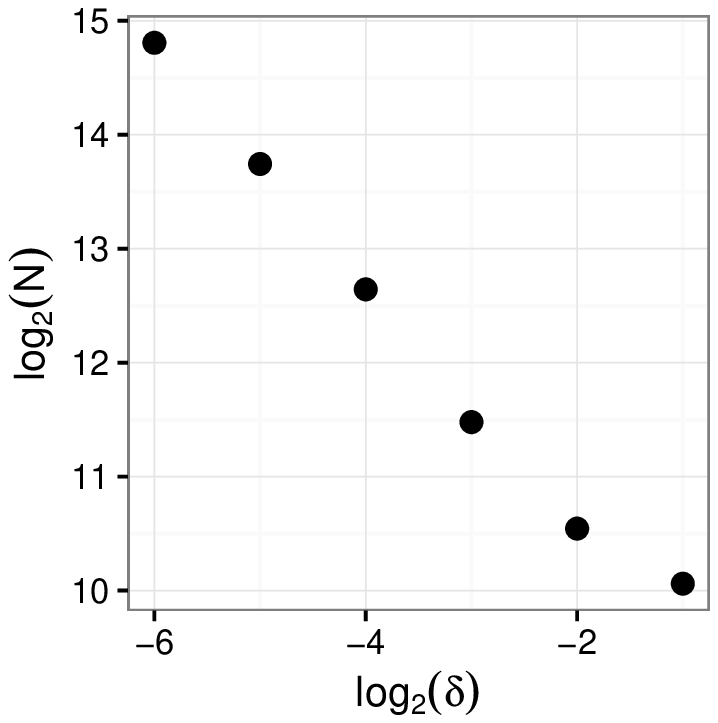}}\caption{Logarithmic plots for the simple adaptive $N$ particle filter estimates
of $\hat{\gamma}_{n}(1)$ for the linear Gaussian example. Figure
(a) plots the sample variance of $\hat{\gamma}_{n}^{N}(1)/\hat{\gamma}_{n}(1)$
against $\delta$, with the straight line $y=x$. Figure (b) plots
$N$ against $\delta$, where $N$ is the average number of particles
used by the final particle filter.\label{fig:lg3}}
\end{figure}

\subsection{Stochastic volatility hidden Markov model}

A stochastic volatility model is defined by $M_{0}(\cdot)=\mathcal{N}\left\{ \,\cdot\,;0,\sigma^{2}/(1-\rho^{2})\right\} $,
$M_{p}(x_{p-1},\cdot)=\mathcal{N}(\,\cdot\,;\rho x_{p-1},\sigma^{2})$
and $G_{p}(x_{p})=\mathcal{N}(y_{p};0,\beta^{2}\exp(x_{p}))$. We
used the pound/dollar daily exchange rates for 100 consecutive weekdays
ending on 28th June, 1985, a subset of the well-known dataset analyzed
in Harvey, Ruiz and Shephard (1994). Our results are obtained by choosing
the parameters $(\rho,\sigma,\beta)=(0.95,0.25,0.5)$. We provide
in the supplement plots of the accuracy of the estimate $N\hat{V}_{n}^{N}(\varphi)$
as $N$ grows using $10^{4}$ replicates for each value of $N$; the
asymptotic variances $\hat{\sigma}_{n}^{2}(1)$ and $\hat{\sigma}_{n}^{2}(Id-\hat{\eta}_{n}(Id))$
are estimated as being approximately $354$ and $1.31$ respectively.
In the supplement we plot the estimates of $\hat{v}_{p,n}(\varphi)$.
We found modest improvement for the approximation of the asymptotically
optimal particle filter, as one could infer from the estimated $\hat{v}_{p,n}(\varphi)$.
For the simple adaptive $N$ particle filter, results are provided
in the supplement, and indicate that the estimates of $\hat{\gamma}_{n}(1)$
are close to their prescribed thresholds.

\subsection{An SMC sampler\label{sub:A-simple-SMC}}

We consider a sequential simulation problem, as described in Section~\ref{sub:Statistical-examples-of},
with $\mathsf{X=\mathbb{R}}$, $\bar{\pi}_{0}(x)=\mathcal{N}(0,10^{2})$
and $\bar{\pi}_{1}(x)=0.3\mathcal{N}(x;-10,0.1^{2})+0.7\mathcal{N}(x;10,0.2^{2})$.
The distribution $\pi_{1}$ is bi-modal with well-separated modes.
With $n=11$, and the sequence of tempering parameters 
\[
\beta_{0:n}=(0,0.0005,0.001,0.0025,0.005,0.01,0.025,0.05,0.1,0.25,0.5,1),
\]
we let each Markov kernel $M_{p}$, $p\in\{1,\ldots,n\}$ be an $\eta_{p}$-invariant
random walk Metropolis kernel iterated $k=10$ times with proposal
variance $\tau_{p}^{2}$, where $\tau_{1:n}=(10,9,8,7,6,5,4,3,2,1,1)$.

One striking difference between the estimates for this model and those
for the hidden Markov models above is that the asymptotic variance
$\sigma_{n}^{2}(Id-\eta_{n}(Id))\approx822$ is considerably larger
than $\sigma_{n}^{2}(1)\approx2.1$; the variability of the estimates
$NV_{n}^{N}(\varphi)$ is shown in the supplement. Inspection of the
estimates of $v_{p,n}(\varphi)$ in Figures~\ref{fig:mixk10} allows
us to investigate both this difference and the dependence of $v_{p,n}(\varphi)$
on $k$ in greater detail.

In Figure~\ref{fig:mixk10}(a)--(b) we can see that while $v_{p,n}(1)$
is small for all $p$, the values of $v_{p,n}(Id-\eta_{n}(Id))$ are
larger for large $p$ than for small $p$; this could be due to the
inability of the Metropolis kernels $(M_{q})_{q\geq p}$ to mix well
due to the separation of the modes in $(\eta_{q})_{q\geq p}$ when
$p$ is large. In Figure~\ref{fig:mixk10}(c)--(d), $k=1$, that
is each $M_{P}$ consists of only a single iterate of a Metropolis
kernel, and we see that the values of $v_{p,n}(\varphi)$ associated
with small $p$ are much larger than when $k=10$, indicating that
the larger number of iterates does improve the asymptotic variance
of the particle approximations. However, the impact on $v_{p,n}(\varphi)$
is less pronounced for large $p$. Results for the simple adaptive
$N$ particle filter approximating $\eta_{n}(Id)$ are provided in
the supplement, which again show that the estimates are close to their
prescribed thresholds.

\begin{figure}
\centering{}\subfloat[$\varphi\equiv1$, $k=10$]{\begin{centering}
\includegraphics[scale=0.45]{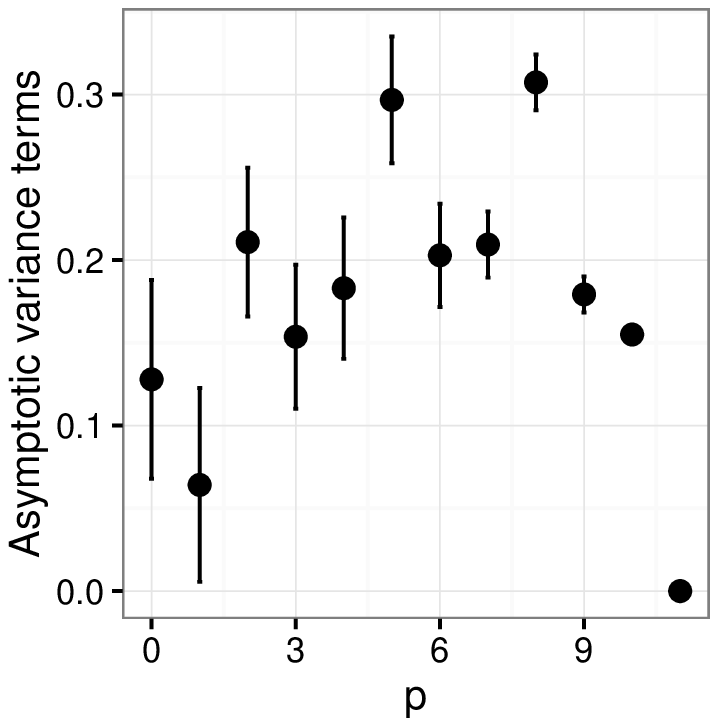}
\par\end{centering}

}\subfloat[$\varphi=Id-\eta_{n}(Id)$, $k=10$]{\centering{}\includegraphics[scale=0.45]{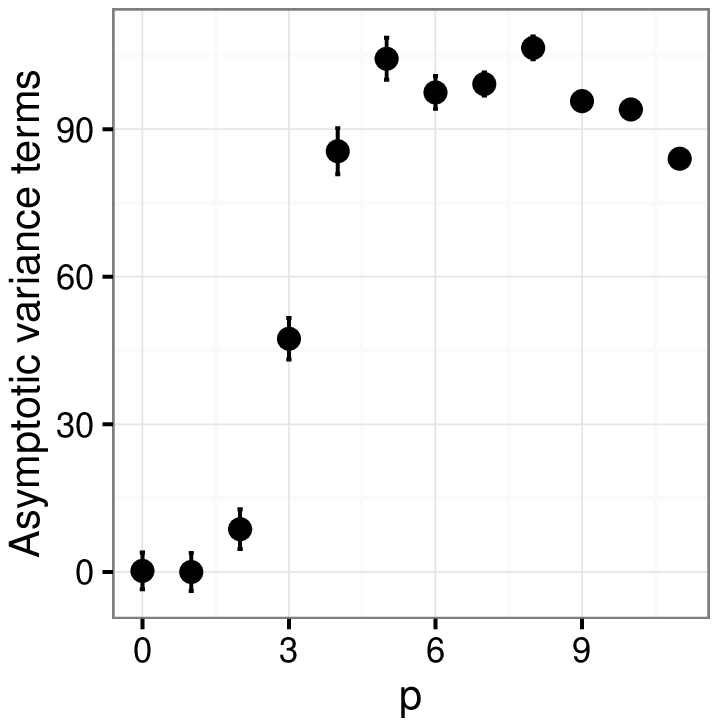}}\subfloat[$\varphi\equiv1$, $k=1$]{\begin{centering}
\includegraphics[scale=0.45]{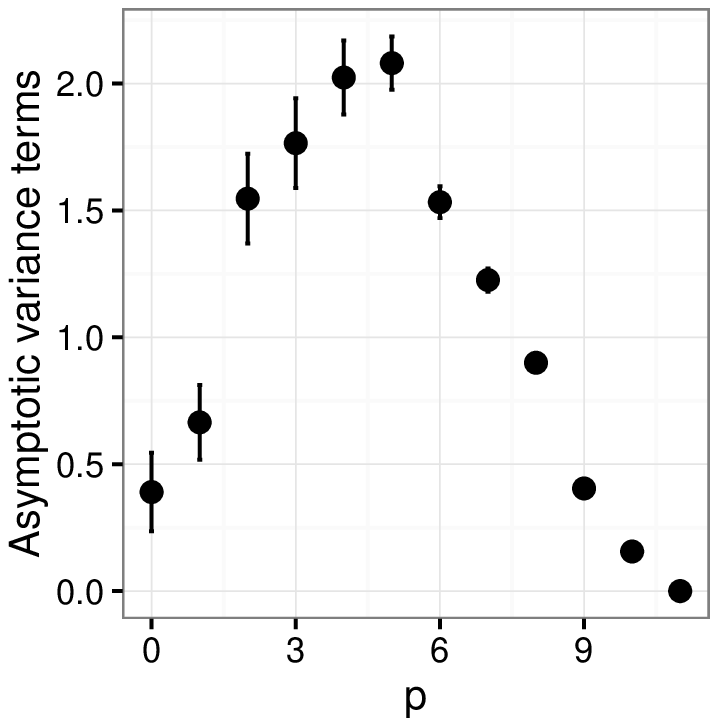}
\par\end{centering}

}\subfloat[$\varphi=Id-\eta_{n}(Id)$, $k=1$]{\centering{}\includegraphics[scale=0.45]{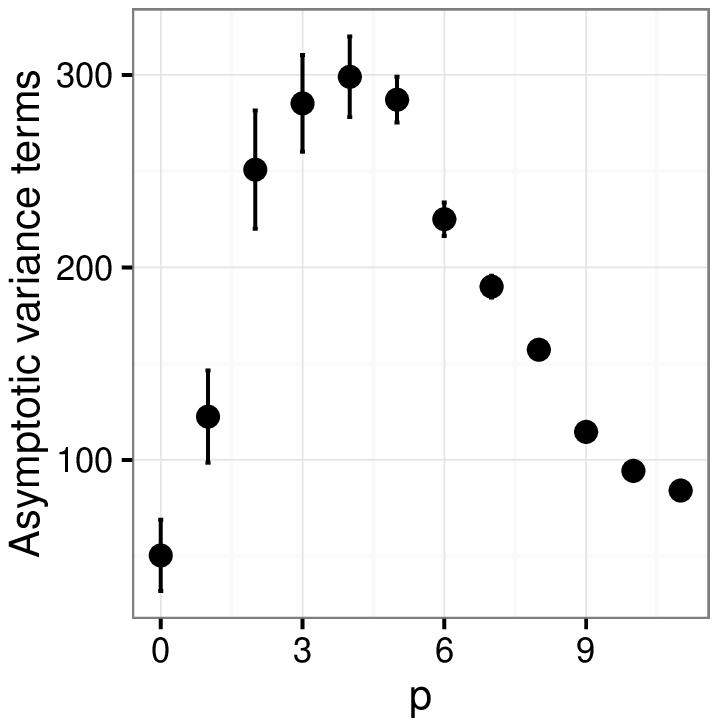}}\caption{Plot of $v_{p,n}^{N}(\varphi)$ (dots and error bars for the mean
$\pm$ one standard deviation) at each $p\in\{0,\ldots,n\}$ with
$k=10$ iterations (a)--(b) and $k=1$ iteration (c)--(d) for each
Markov kernel in the SMC sampler example and $N=10^{5}$.\label{fig:mixk10}}
\end{figure}

\section{Discussion}

\subsection{Alternatives to the bootstrap particle filter}

In the hidden Markov model examples above, we have constructed the
Feynman--Kac measures taking $M_{0},\ldots,M_{n}$ to be the initial
distribution and transition probabilities of the latent process and
defining $G_{0},\ldots,G_{n}$ to incorporate the realized observations.
This is only one, albeit important, way to construct particle approximations
of $\eta_{n}$, and the algorithm itself is usually referred to as
the bootstrap particle filter. Alternative specifications of $(M_{p},G_{p})_{0\leq p\leq n}$
lead to different Feynman-Kac models, as discussed in \citet[Section~2.4.2]{DelMoral2004},
and the variance estimators introduced here are applicable to these
models as well.

One particular specification corresponds to the ``fully adapted''
auxiliary particle filter of \citet{pitt1999filtering}, as discussed
by \citet{Doucet2008}. Specifically, we define $\check{M}_{0}({\rm d}x_{0})=M_{0}({\rm d}x_{0})G_{0}(x_{0})/M_{0}(G_{0})$,
and 
\[
\check{M}_{p}(x_{p-1},{\rm d}x_{p})=\frac{M_{p}(x_{p-1},{\rm d}x_{p})G_{p}(x_{p})}{M_{p}(G_{p})(x_{p-1})},\qquad p\in[n],
\]
and then $\check{G}_{0}(x_{0})=M_{0}(G_{0})M_{1}(G_{1})(x_{0})$ and
$\check{G}{}_{p}(x_{p})=M_{p+1}(G_{p+1})(x_{p})$, $p\geq1$. If we
denote by $\check{\gamma}_{n}$ and $\check{\eta}_{n}$ the Feynman--Kac
measures associated with $(\check{M}_{p},\check{G}_{p})_{0\leq p\leq n}$,
we obtain $\check{\gamma}_{n}=\hat{\gamma}_{n}$ and $\check{\eta}_{n}=\hat{\eta}_{n}$.
Moreover, the variances of $\check{\gamma}_{n}^{N}(\varphi)$ and
$\check{\eta}_{n}^{N}(\varphi)$ are often smaller than the variances
of $\hat{\gamma}_{n}^{N}(\varphi)$ and $\hat{\eta}_{n}^{N}(\varphi)$.
In Figure~\ref{fig:lg_apf}, we plot the corresponding $\check{v}_{p,n}(1)$
and their approximations for the same linear Gaussian example in Section~\ref{sub:Linear-Gaussian-HMM}.
Here, the asymptotic variance of $\check{\gamma}_{n}^{N}(1)/\check{\gamma}_{n}(1)$
is $40.718$, more than $7$ times smaller than $\hat{\sigma}_{n}^{2}(1)$.

\begin{figure}
\begin{centering}
\includegraphics[scale=0.6]{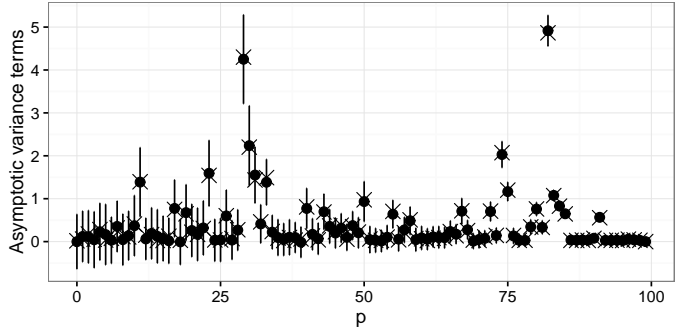}
\par\end{centering}

\centering{}\caption{Plot of $\check{v}_{p,n}^{N}(1)$ (dots and error bars for the mean
$\pm$ one standard deviation) and $\check{v}_{p,n}(1)$ (crosses)
at each $p\in\{0,\ldots,n\}$ in the Linear Gaussian example.\label{fig:lg_apf}}
\end{figure}

\subsection{Estimators based on i.i.d. replicates}

It is clearly possible to estimate consistently the variance of $\gamma_{n}^{N}(\varphi)/\gamma_{n}(1)$
by using i.i.d. replicates of $\gamma_{n}^{N}$. Such estimates necessarily
entail simulation of multiple particle filters. We now compare the
accuracy of such estimates with those based on i.i.d. replicates of
$V_{n}^{N}(\varphi)$. For some $\varphi\in\mathcal{L}(\mathcal{X})$
and $B\in\mathbb{N}$, let $\gamma_{n,i}^{N}(\varphi)$ and $V_{n,i}^{N}(\varphi)$
be i.i.d. replicates for $i\in[B]$, and define $M=N^{-1}\sum_{i\in[B]}\gamma_{n,i}^{N}(1)$.
A standard estimate of ${\rm var}\left\{ \gamma_{n}^{N}(\varphi)/\gamma_{n}(1)\right\} $
is obtained by calculating the sample variance of $\{M^{-1}\gamma_{n,i}^{N}(\varphi)\:;\:i\in[B]\}$.
Noting the lack-of-bias of $\gamma_{n}^{N}(1)^{2}V_{n}^{N}(\varphi)$,
an alternative estimate of ${\rm var}\left[\gamma_{n}^{N}(\varphi)/\gamma_{n}(1)\right]$
can be obtained as $\frac{1}{B}\sum_{i\in[B]}[M^{-1}\gamma_{n,i}^{N}(1)]^{2}V_{n,i}^{N}(\varphi)$.
Both these estimates can be seen as ratios of simple Monte Carlo estimates
of ${\rm var}\left\{ \gamma_{n}^{N}(\varphi)\right\} $ and $\gamma_{n}(1)^{2}$,
and are therefore consistent as $B\rightarrow\infty$. We show in
Figure~\ref{fig:replicates} a comparison between these estimates
for the three models discussed in Section~\ref{sec:Numerical-illustrations}
with $N_{}=10^{3}$ and $\varphi\equiv1$, and we can see that the
alternative estimate based on $\hat{V}_{n}^{N}(1)$ is empirically
more accurate for these examples.

\begin{figure}
\begin{centering}
\includegraphics[scale=0.5]{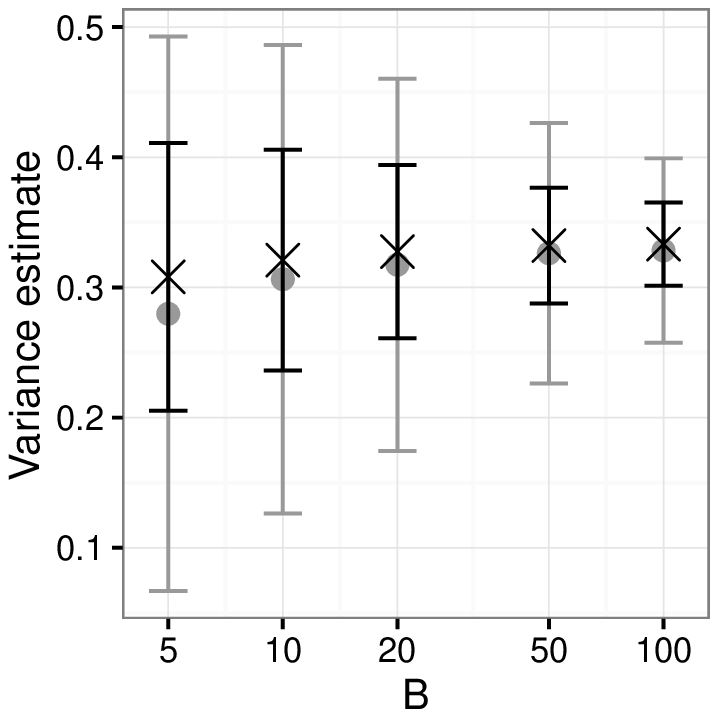}\includegraphics[scale=0.5]{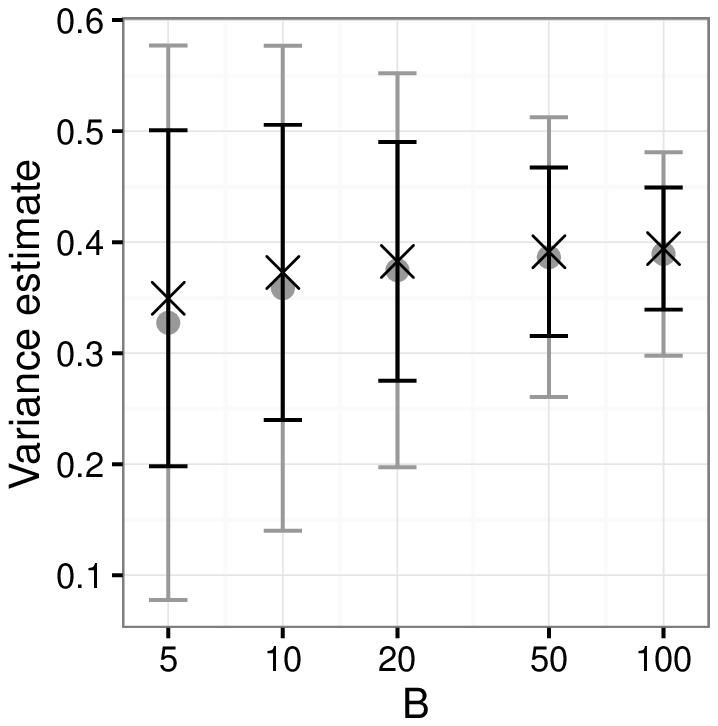}\includegraphics[scale=0.5]{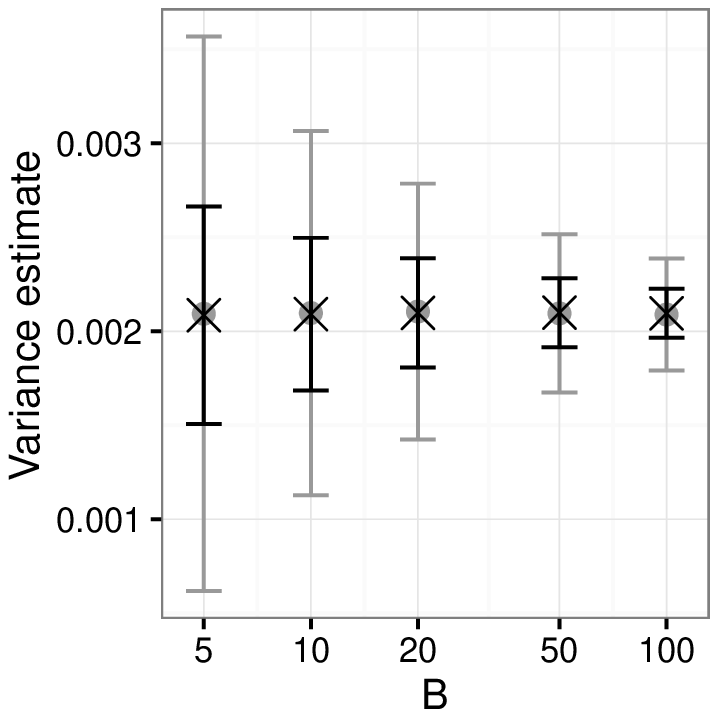}
\par\end{centering}

\caption{Plot of the standard estimate of ${\rm var}\left[\hat{\gamma}_{n}^{N}(\varphi)/\hat{\gamma}_{n}(1)\right]$
(gray dots and error bars) and the alternative estimate using $\hat{V}_{n}^{N}(1)$
(black crosses and error bars) against $B$ in (left to right) the
examples of Sections~\ref{sub:Linear-Gaussian-HMM}--\ref{sub:A-simple-SMC}.\label{fig:replicates} }
\end{figure}

\subsection{Final remarks}

The particular approximations developed here provide a natural way
to estimate the terms appearing in the non-asymptotic second moment
expression (\ref{eq:second_moment_formula}). To the best of our knowledge,
we have also provided the first generally applicable, consistent estimators
of $v_{p,n}(\varphi)$. The expression (\ref{eq:second_moment_formula})
does not apply to particle approximations with resampling schemes
other than multinomial, and one possible avenue of future research
is to investigate estimators in these other settings. Whilst we have
emphasized variances and asymptotic variances, the measures $\mu_{b}$
also appear in expressions which describe propagation of chaos properties
of the particle system. For instance, in the situation $N_{p}\equiv N$,
the asymptotic bias formula of \citet[p.7.]{del2007sharp} can be
expressed as
\[
NE\left\{ \eta_{n}^{N}(\varphi)-\eta_{n}(\varphi)\right\} \rightarrow-\sum_{p=0}^{n-1}\frac{\eta_{p}\left\{ Q_{p,n}(1)Q_{p,n}(\varphi-\eta_{n}(\varphi))\right\} }{\eta_{p}Q_{p,n}(1)^{2}}\equiv-\sum_{p=0}^{n-1}\frac{\mu_{e_{p}}\left\{ 1\otimes(\varphi-\eta_{n}(\varphi))\right\} }{\gamma_{n}(1)^{2}},
\]
which could be consistently estimated by replacing $\mu_{e_{p}}$
and $\gamma_{n}(1)$ by $\mu_{e_{p}}^{N}$ and $\gamma_{n}^{N}(1)$.
Finally, the technique used in the proof of Lemma~\ref{lem:Q2_change}
can be generalized to obtain expressions for arbitrary positive integer
moments of $\gamma_{n}^{N}(\varphi)$.

\section*{Supplementary Material}

The supplementary material at \url{http://www.warwick.ac.uk/alee/vestpf_supp.pdf}
includes algorithms for efficient computation of the variance estimators,
and proofs of Corollary 1, Lemmas 1--3, Propositions 1--2, and Theorems
2, 4 and 5.

\appendix

\section*{Appendix}

\begin{proof}[Proof of Theorem~\ref{thm:V_n^N_thm_front}]Throughout
the proof, $\to$ denotes convergence in probability. For part 1.,
the fact $\mu_{0_{n}}(\varphi^{\otimes2})=\gamma_{n}(\varphi)^{2}$
and Theorem \ref{THM:MU_ B^N} together give 
\[
E\left\{ \gamma_{n}^{N}(1)^{2}V_{n}^{N}(\varphi)\right\} =E\left\{ \gamma_{n}^{N}(\varphi)^{2}-\mu_{0_{n}}^{N}(\varphi^{\otimes2})\right\} =E\left\{ \gamma_{n}^{N}(\varphi)^{2}\right\} -\gamma_{n}(\varphi)^{2}={\rm var}\left\{ \gamma_{n}^{N}(\varphi)\right\} .
\]
For part 2., combining the identity (\ref{eq:gamma_n^N_mu_identity}),
$\mu_{b}^{N}(\varphi^{\otimes2})\to\mu_{b}(\varphi^{\otimes2})$ by
Theorem~\ref{THM:MU_ B^N}, and the fact that for any $b\in B_{n}$
other than $0_{n}$ and $e_{0},\ldots,e_{n}$, $\prod_{p=0}^{n}\left(\frac{1}{N_{p}}\right)^{b_{p}}\left(1-\frac{1}{N_{p}}\right)^{1-b_{p}}$
is in $\mathcal{O}(N^{-2})$, we obtain
\begin{equation}
\gamma_{n}^{N}(\varphi)^{2}-\mu_{0_{n}}^{N}(\varphi^{\otimes2})=\left\{ \sum_{p=0}^{n}\frac{\mu_{e_{p}}^{N}(\varphi^{\otimes2})-\mu_{0_{n}}^{N}(\varphi^{\otimes2})}{\left\lceil c_{p}N\right\rceil }\right\} +\mathcal{O}_{p}(N^{-2}).\label{eq:O_p_bound-1}
\end{equation}
Also noting that by Proposition~\ref{prop:asconv-1} $\gamma_{n}^{N}(1)^{2}\to\gamma_{n}(1)^{2}$,
from (\ref{eq:vpndefn}) that $\gamma_{n}(1)^{2}v_{p,n}(\varphi)=\mu_{e_{p}}(\varphi^{\otimes2})-\mu_{0_{n}}(\varphi^{\otimes2})$
and again using $\mu_{b}^{N}(\varphi^{\otimes2})\to\mu_{b}(\varphi^{\otimes2})$,
we then have 
\begin{equation}
NV_{n}^{N}(\varphi)=\frac{N}{\gamma_{n}^{N}(1)^{2}}\left\{ \gamma_{n}^{N}(\varphi)^{2}-\mu_{0_{n}}^{N}(\varphi^{\otimes2})\right\} \to\sum_{p=0}^{n}\frac{v_{p,n}(\varphi)}{c_{p}}=\sigma_{n}^{2}(\varphi).\label{eq:V_n_to_sig_gamma-1}
\end{equation}
For part 3., first note that by Theorem~\ref{THM:MU_ B^N} and Proposition~\ref{prop:asconv-1},
for any $b\in B_{n}$,
\begin{eqnarray*}
\mu_{b}^{N}([\varphi-\eta_{n}^{N}(\varphi)]^{\otimes2}) & = & \mu_{b}^{N}(\varphi^{\otimes2})-\eta_{n}^{N}(\varphi)[\mu_{b}^{N}(\varphi\otimes1)+\mu_{b}^{N}(1\otimes\varphi)]+\eta_{n}^{N}(\varphi)^{2}\mu_{b}^{N}(1^{\otimes2})\\
 & \to & \mu_{b}([\varphi-\eta_{n}(\varphi)]^{\otimes2}),
\end{eqnarray*}
from which it follows that (\ref{eq:O_p_bound-1}) also holds with
$\varphi$ replaced by \textbf{$\varphi-\eta_{n}^{N}(\varphi)$},
and then
\[
NV_{n}^{N}(\varphi-\eta_{n}^{N}(\varphi))\to\sum_{p=0}^{n}\frac{v_{p,n}(\varphi-\eta_{n}(\varphi))}{c_{p}}=\sigma_{n}^{2}(\varphi-\eta_{n}(\varphi)),
\]
similarly to (\ref{eq:V_n_to_sig_gamma-1}). \end{proof}

\begin{proof}[Proof of Lemma~\ref{lem:eve_identical_events}]For $i\in[N_{n}]$
define $B_{n-1}^{i}=A_{n-1}^{i}$ and $B_{p-1}^{i}=A_{p-1}^{B_{p}^{i}}$
for $p\in[n-1]$. Since in Algorithm \ref{alg:bpf}, $E_{p}^{i}=E_{p-1}^{A_{p-1}^{i}}$
for all $p\in[n],i\in[N_{p}]$ , a simple inductive argument then
shows that 
\begin{equation}
E_{n}^{i}=E_{p}^{B_{p}^{i}},\quad p\in\{0,\dots,n\},\,i\in[N_{n}].\label{eq:E_n=00003DE_p^B}
\end{equation}
We shall now prove $(K^{1},K^{2})\in\mathcal{I}(0_{n})\Rightarrow E_{n}^{K_{n}^{1}}\neq E_{n}^{K_{n}^{2}}$.
Recall from Section~\ref{sub:Genealogical-tracing-variables} that
when $(K^{1},K^{2})\in\mathcal{I}(0_{n})$, we have $A_{p-1}^{K_{p}^{1}}=K_{p-1}^{1}\neq K_{p-1}^{2}=A_{p-1}^{K_{p}^{2}}$
for all $p\in[n]$, hence $B_{0}^{K_{n}^{1}}=K_{0}^{1}\neq K_{0}^{2}=B_{0}^{K_{n}^{2}}$.
Applying (\ref{eq:E_n=00003DE_p^B}) with $p=0$ and using the fact
that in Algorithm~\ref{alg:bpf}, $E_{0}^{i}=i$ for all $i\in[N_{n}]$,
we have $E_{n}^{i}=E_{0}^{B_{0}^{i}}=B_{0}^{i}$, hence $E_{n}^{K_{n}^{1}}=B_{0}^{K_{n}^{1}}\neq B_{0}^{K_{n}^{2}}=E_{n}^{K_{n}^{2}}$
as required. It remains to prove $(K^{1},K^{2})\notin\mathcal{I}(0_{n})\Rightarrow E_{n}^{K_{n}^{1}}=E_{n}^{K_{n}^{2}}$.
Assuming $(K^{1},K^{2})\notin\mathcal{I}(0_{n})$, consider $\tau=\mathrm{max}\{p:K_{p}^{1}=K_{p}^{2}\}$.
If $\tau=n$ then clearly $E_{n}^{K_{n}^{1}}=E_{n}^{K_{n}^{2}}$,
so suppose $\tau<n$. It follows from Section \ref{sub:Genealogical-tracing-variables}
that $B_{\tau}^{K_{n}^{1}}=K_{\tau}^{1}=K_{\tau}^{2}=B_{\tau}^{K_{n}^{2}}$,
so taking $p=\tau$ and $i=K_{n}^{1},K_{n}^{2}$ in (\ref{eq:E_n=00003DE_p^B})
gives $E_{n}^{K_{n}^{1}}=E_{n}^{K_{n}^{2}}$.  \end{proof}

\begin{proof}[Proof of Theorem~\ref{thm:v_pn^N}]For part 1., Theorem
\ref{THM:MU_ B^N} gives
\[
E\left\{ \gamma_{n}^{N}(1)^{2}v_{p,n}^{N}(\varphi)\right\} =E\left\{ \mu_{e_{p}}^{N}(\varphi^{\otimes2})-\mu_{0_{n}}^{N}(\varphi^{\otimes2})\right\} =\mu_{e_{p}}(\varphi^{\otimes2})-\mu_{0_{n}}(\varphi^{\otimes2})=\gamma_{n}(1)^{2}v_{p,n}(\varphi).
\]
For the remainder of the proof, $\to$ denotes convergence in probability.
For part 2., $\mu_{e_{p}}^{N}(\varphi^{\otimes2})-\mu_{0_{n}}^{N}(\varphi^{\otimes2})\to\gamma_{n}(1)^{2}v_{p,n}(\varphi)$
by Theorem~\ref{THM:MU_ B^N}, and $\gamma_{n}^{N}(1)^{2}\to\gamma_{n}(1)^{2}$
by Proposition~\ref{prop:asconv-1}, so $v_{p,n}^{N}(\varphi)=\left[\mu_{e_{p}}^{N}(\varphi^{\otimes2})-\mu_{0_{n}}^{N}(\varphi^{\otimes2})\right]/\gamma_{n}^{N}(1)^{2}\to v_{p,n}(\varphi)$;
as in the proof of Theorem~\ref{thm:V_n^N_thm_front}, $\mu_{b}^{N}([\varphi-\eta_{n}^{N}(\varphi)]^{\otimes2})\to\mu_{b}([\varphi-\eta_{n}(\varphi)]^{\otimes2})$
gives $v_{p,n}^{N}(\varphi-\eta_{n}^{N}(\varphi))\to v_{p,n}(\varphi-\eta_{n}(\varphi))$.
Part 3. follows from parts 1. and 2. \end{proof}

\bibliographystyle{abbrvnat}
\bibliography{avarest}

\end{document}